\documentclass[11pt]{article}

\usepackage{amssymb,amsmath,amsthm,amsfonts,dsfont}
\usepackage[numbers,comma,sort&compress]{natbib}
\usepackage[colorlinks,hypertexnames=false]{hyperref}
\hypersetup{
	pdfstartview={FitH},
	pdfnewwindow=true,
	colorlinks=true,
	linkcolor=blue,
	citecolor=blue,
	filecolor=blue,
	urlcolor=blue}
\usepackage[capitalise]{cleveref}
\usepackage[all]{hypcap}
\usepackage{url}
\usepackage{graphicx}
\usepackage[font=small]{caption}
\usepackage[subrefformat=parens,labelformat=parens]{subcaption}
\usepackage{float}
\usepackage{footnote}
\usepackage{enumitem}
\usepackage[margin=1in]{geometry}

\newcommand{\abs}[1]{\left\lvert#1\right\rvert}
\newcommand{\norm}[1]{\left\lVert#1\right\rVert}

\DeclareMathOperator{\polylog}{polylog}

\newtheorem{theorem}{Theorem}
\newtheorem{definition}{Definition}
\newtheorem{lemma}{Lemma}
\newtheorem{proposition}[lemma]{Proposition}
\newtheorem{corollary}[lemma]{Corollary}
\theoremstyle{remark}

\theoremstyle{plain}
\newtheorem*{definition*}{Definition}

\newcommand{\thm}[1]{\hyperref[thm:#1]{Theorem~\ref*{thm:#1}}}
\newcommand{\defn}[1]{\hyperref[defn:#1]{Definition~\ref*{defn:#1}}}
\newcommand{\lem}[1]{\hyperref[lem:#1]{Lemma~\ref*{lem:#1}}}
\newcommand{\prop}[1]{\hyperref[prop:#1]{Proposition~\ref*{prop:#1}}}
\newcommand{\fig}[1]{\hyperref[fig:#1]{Figure~\ref*{fig:#1}}}
\newcommand{\tab}[1]{\hyperref[tab:#1]{Table~\ref*{tab:#1}}}
\renewcommand{\sec}[1]{\hyperref[sec:#1]{Section~\ref*{sec:#1}}}
\newcommand{\append}[1]{\hyperref[append:#1]{Appendix~\ref*{append:#1}}}
\newcommand{\cor}[1]{\hyperref[cor:#1]{Corollary~\ref*{cor:#1}}}
\newcommand{\obs}[1]{\hyperref[obs:#1]{Observation~\ref*{obs:#1}}}

\newcommand{\ket}[1]{|#1\rangle}
\newcommand{\bra}[1]{\langle#1|}
\newcommand{\ketbra}[2]{\ket{#1}\!\bra{#2}}

\usepackage{mathtools}

\usepackage{nicematrix}
\usepackage{tikz}

\usepackage[hyperpageref]{backref}
\renewcommand*{\backrefalt}[4]{%
\ifcase #1 %
No citations.%
\or
(Cited on page #2).%
\else
(Cited on pages #2).%
\fi
}

\usepackage{etoolbox}
\makeatletter
\patchcmd\NAT@citexnum{\let\NAT@last@num\NAT@num}{\MakeLinkTarget[cite]{}\Hy@backout{\@citeb\@extra@b@citeb}\let\NAT@last@num\NAT@num}{}{\fail}
\makeatother

\usepackage{etoolbox}
\apptocmd{\sloppy}{\hbadness 10000\relax}{}{}

\usepackage{authblk}

\usepackage{quantikz}
\usetikzlibrary{fit,backgrounds,arrows.meta,calc,patterns,patterns.meta}

\usepackage{standalone}
\usepackage{pgfplots}
\pgfplotsset{compat=1.18}
\title{Quantum oblique eigenprojection}
\author{Alexander M.\ Dalzell}
\author{Yuan Su}
\affil{AWS Center for Quantum Computing, Pasadena, CA 91106, USA}

\date{}

\begin{document}
\maketitle

%%%%%%%%%%%%%%%%%%%%%%%%%%%%%%%%%%%%%%%%%%%%%%%%%%%%%%%%%%%%%%%%%%%%%%%%%%%%%%
% % Short abstract
% Every square matrix decomposes its underlying Hilbert space into generalized, nonorthogonal eigensubspaces. We show that a quantum computer can perform such an oblique eigenprojection $\Pi$ given block encoding access to the input matrix. Our approach has a query complexity nearly linear in the inverse gap and a normalization factor close to $\lVert\Pi\rVert$ under a spectral-set condition on the input. This covers common assumptions on the numerical range or diagonalizability and matches known results for orthogonal eigenprojections. We achieve this with a two-sided block preconditioning that uses a discrete Fourier transform of the matrix resolvent. We describe applications to: (i)~preparing eigenstates of matrices with complex eigenvalues, extending the quantum eigenvalue transformation algorithm of Low and Su beyond real spectra; (ii)~solving continuous-time algebraic Riccati equations, cubically speeding up a prior solver of Rodenas-Ruiz, Zhao, and Lee; and (iii)~solving ordinary Sylvester equations, quadratically improving a direct augmented method of Wang and Liu. Our result suggests a promising route to applying nonanalytic matrix functions on quantum computers.

%%%%%%%%%%%%%%%%%%%%%%%%%%%%%%%%%%%%%%%%%%%%%%%%%%%%%%%%%%%%%%%%%%%%%%%%%%%%%%
\begin{abstract}
Every square matrix decomposes its underlying Hilbert space into generalized eigensubspaces. These subspaces are typically nonorthogonal, and performing their projections amounts to applying a discontinuous, and hence nonanalytic, function of a complex variable to the input matrix.

We show that a quantum computer can perform such an oblique eigenprojection $\Pi$ given block encoding access to the input matrix. Our approach has a query complexity nearly linear in the inverse gap and a normalization factor close to $\lVert\Pi\rVert$ under a spectral-set condition on the input, satisfied in particular under common assumptions on its numerical range or diagonalizability.

Our main technical contribution is a two-sided block preconditioning that employs a discrete Fourier transform of the matrix resolvent. This achieves the desired normalization factor while maintaining the same query complexity as the standard block encoding approach based on resolvent integration, matching known results for orthogonal eigenprojections.

We describe applications to: (i)~preparing eigenstates of matrices with complex eigenvalues, extending the quantum eigenvalue transformation algorithm of Low and Su beyond real spectra; (ii)~solving continuous-time algebraic Riccati equations, cubically speeding up a prior solver of Rodenas-Ruiz, Zhao, and Lee; and (iii)~solving ordinary Sylvester equations, quadratically improving a direct augmented method of Wang and Liu. 

Our result suggests a promising route to applying nonanalytic matrix functions on quantum computers.
\end{abstract}

%%%%%%%%%%%%%%%%%%%%%%%%%%%%%%%%%%%%%%%%%%%%%%%%%%%%%%%%%%%%%%%%%%%%%%%%%%%%%%
\newpage
{
	\thispagestyle{empty}
	\clearpage\tableofcontents
	\thispagestyle{empty}
}
\newpage

%%%%%%%%%%%%%%%%%%%%%%%%%%%%%%%%%%%%%%%%%%%%%%%%%%%%%%%%%%%%%%%%%%%%%%%%%%%%%% Main text

%%%%%%%%%%%%%%%%%%%%%%%%%%%%%%%%%%%%%%%%%%%%%%%%%%%%%%%%%%%%%%%%%%%%%%%%%%%%%%
\section{Introduction}
\label{sec:intro}
Quantum computers hold the promise of solving numerous computational problems faster than their classical counterparts~\cite{shor1997, grover1996}. A prominent example is the ground state projection problem, whose solution gives access to the static properties of many-body quantum systems relevant to application areas like chemistry, nuclear physics, and materials science. In this problem, we are given a description of a many-body Hamiltonian $A$ along with an assumption on the gap $\delta$ between the ground energy and the first excited energy, and the goal is to efficiently perform the eigenprojection $\Pi_{\lambda_0}$ onto the eigenstate associated with the ground energy $\lambda_0$. One solution to this problem---known already in the early days of quantum algorithms research---applies the phase estimation algorithm to the time evolution unitary $e^{iAt}$, effecting the projection $\Pi_{\lambda_0}$ \cite{kitaev1995,cleve1997,abrams1999}. To resolve the ground and excited states this way, the evolution must be performed out to time $t=\operatorname{\pmb{\Theta}}(1/\delta)$---a fundamental lower bound related to the energy--time uncertainty principle.

Modern approaches to eigenprojection for Hermitian matrices indeed achieve optimal total query cost $\operatorname{\pmb{\Theta}}(1/\delta)$ \cite{ge2019,lin2020}, but they are better viewed through the lens of linear algebra: one starts with a description of the input matrix $A$ and aims to perform a particular matrix function $f(A)$ on a quantum computer.  When $A$ is Hermitian, its singular values agree with eigenvalues in magnitude, whereas its eigenstates form an orthonormal basis. Then the task can be accomplished by the Quantum Singular Value Transformation (QSVT)~\cite{gilyen2019}. Crucially, QSVT can perform polynomial transformations to Hermitian input matrices with a query complexity linear in the degree of the polynomial, and thus can remain efficient even when the underlying Hilbert space has an exponentially large dimension. For example, QSVT optimally solves the orthogonal eigenprojection problem by using a degree-$\operatorname{\mathbf{O}}(1/\delta)$ polynomial $f$ that approximates the indicator function on the target part of the spectrum, so that $f(A) \approx \Pi_{\lambda_0}$ \cite[Lemma 29]{gilyen2019}. But QSVT applies far more widely, providing a unified framework for developing many quantum algorithms~\cite{martyn2021} with nearly optimal query complexity~\cite{Montanaro2024quantum}, including Hamiltonian simulation~\cite{LC17}, Gibbs state preparation, and solving linear systems \cite{gilyen2019}, which correspond to different choices of the function $f$. 

Performing $f(A)$ for a general, non-normal matrix $A$ is considerably harder, as the eigenvalues and singular values of $A$ no longer relate and the associated generalized eigensubspaces are typically oblique. Nonetheless, the ability to implement Quantum EigenValue Transformation (QEVT) for non-normal operators is broadly useful: it arises in simulating non-Hermitian physics~\cite{Ashida20}, in transcorrelated approaches to quantum chemistry~\cite{McArdle20}, in solving differential equations~\cite{Berry2017Differential,BerryCosta22,Krovi2023improvedquantum}, and in studying irreversible Markov chains~\cite{banerjee2026}. When the target function $f$ is analytic, several approaches have been developed to transform non-normal matrices on a quantum computer, which can be broadly grouped into three methodologies: those based on resolvent integration~\cite{Takahira2020QuantumCauchy,Takahira21, jiangan2026,wang2026psf,2021Yupreconditioned}, those based on solving a system of linear equations built from a matrix generating function~\cite{QEVP}, and those based on a linear combination of Hamiltonian simulations or a linear combination of Hermitian matrices~\cite{niying2026, wang2026lchm,laplace_qevt}.

An even more challenging task is to perform \emph{nonanalytic} functions of general matrices. A simple yet fundamental example is the eigenprojection $\Pi$ associated with a target region of the spectrum in the complex plane, corresponding to the indicator function of that region---a natural generalization of the influential ground state projection problem discussed above. As illustrated in \fig{indicator}, a polynomial can approximate this indicator function when the spectrum is real, but doing so is fundamentally obstructed when the spectrum is complex~\cite[Corollary B.24]{knapp2016basic}\cite[Theorem 11.5.5]{humpherys2017}, where it is unclear how to polynomially approximate the indicator function in every direction~\cite{cook2020}. Such oblique eigenprojections are nonetheless useful in a variety of settings, including preparing eigenstates of non-normal matrices, solving Sylvester equations, and solving Riccati equations arising in quantum chemistry~\cite{WL2026, RZL2026}.

\begin{figure}[t]
\centering
\begin{subfigure}[b]{0.5\textwidth}
  \centering
  \includegraphics{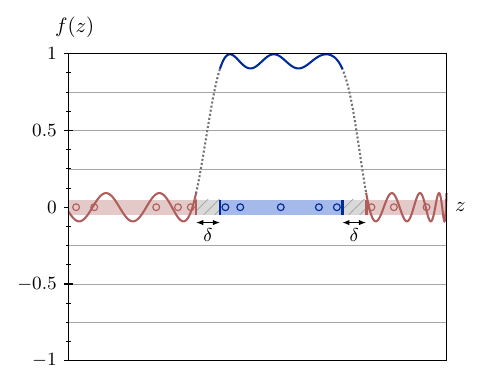}
  \caption{Real spectrum}
  \label{fig:indicator-poly}
\end{subfigure}%
\hfill
\begin{subfigure}[b]{0.5\textwidth}
  \centering
  \includegraphics{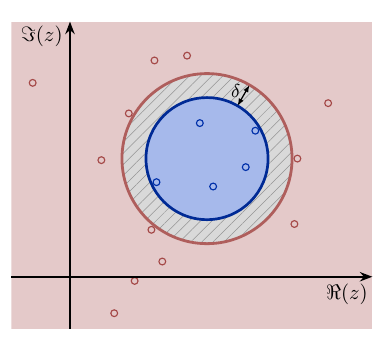}
  \caption{Complex spectrum}
  \label{fig:target-region}
\end{subfigure}
\caption{The eigenprojection problem asks to project onto generalized eigensubspaces associated with eigenvalues (open circles) lying within a particular target region (blue region), promised that all other eigenvalues lie at least $\delta$-far (red region) from the target region. We illustrate the target region as a disk, but we also consider more general regions. (a) If all eigenvalues are real, there exist suitable polynomials $f$ with degree linear in $1/\delta$ which approximate the indicator function, that is, a function which is 1 on the blue region, 0 on the red region, and arbitrary in the gray region where there are no eigenvalues. This enables optimal (in $1/\delta$) eigenprojection via QSVT or QEVT. (b) The same does not hold when the spectrum is complex. It is unclear how to construct a polynomial (or even analytic) function which approximates the analogous indicator function on the blue and red regions.  }
\label{fig:indicator}
\end{figure}

Resolvent integrals offer a way around this obstruction of nonanalyticity. The strategy is to select a contour enclosing the target spectrum and integrate the resolvent along it, which recovers the oblique eigenprojection through Cauchy's integral formula for matrices. This resolvent-integration approach was first used to estimate the eigenvalue density in a target region~\cite{Futamura2021}, and was subsequently developed within the block encoding framework and analyzed in detail in~\cite{RZL2026}. In the special case where the target region is a half-plane, one can instead integrate the resolvent along the axis~\cite{WL2026}. However, a drawback with this approach is that it introduces a normalization factor $\alpha=\operatorname{\pmb{\Theta}}(\norm{\Pi}/\delta)$ with $\delta$ the gap between the spectrum of $A$ and boundary of the target region, which is much larger than the desired normalization factor $\norm{\Pi}$. This is in stark contrast to the QSVT and QEVT frameworks for analytic functions, where $f(A)$ can be implemented with a normalization factor as small as $\norm{f(A)}$. The inflated normalization factor significantly increases the cost of the resulting quantum algorithms when the eigenprojection is used as a subroutine within a larger computation---as in the solution of Riccati equations that we discuss in this paper.

In this work, we show that a quantum computer can perform the oblique eigenprojection $\Pi$ given block encoding access to the input matrix $A$. Our approach has a query complexity nearly in $1/\delta$, and it produces a block encoding with a normalization factor close to $\alpha=\operatorname{\pmb{\Theta}}(\norm{\Pi})$ under a spectral-set condition on $A$, covering common assumptions on its numerical range or diagonalizability. As $\norm{\Pi}$ is the smallest factor achievable by any block encoding of $\Pi$, this normalization is necessarily nearly optimal. Compared with the standard resolvent integration approach, our result quadratically improves the cost of oblique eigenprojection as a subroutine in a larger computation, and it matches analogous results previously established only for the orthogonal case~\cite{lin2020} and real spectrum~\cite{QEVP}.

Our main technical contribution is a two-sided block preconditioning that uses a discrete Fourier transform of the matrix resolvent. Whereas the standard resolvent integral extracts only the top-left corner carrying the projector, we characterize every block through the discrete Fourier transform of the resolvent. This exposes a circulant structure of the block encoding, and preconditioning against it suppresses the inverse-gap growth of the normalization factor down to nearly $\norm{\Pi}$, while maintaining the same query complexity as the standard approach.

As the oblique eigenprojection serves as a fundamental primitive in the design of quantum algorithm, our result can be used to extend and improve a variety of prior results. We describe three such applications. First, we prepare eigenstates of matrices with complex eigenvalues, extending the QEVT algorithm of Low and Su~\cite{QEVP} beyond real spectra. Second, we solve continuous-time algebraic Riccati equations, cubically improving a prior solver of Rodenas-Ruiz, Zhao, and Lee~\cite{RZL2026}. Third, we solve ordinary Sylvester equations, quadratically improving a direct augmented method of Wang and Liu~\cite{WL2026}. These applications suggest a promising route toward applying more general nonanalytic matrix functions on quantum computers.

%%%%%%%%%%%%%%%%%%%%%%%%%%%%%%%%%%%%%%%%%%%%%%%%%%%%%%%%%%%%%%%%%%%%%%%%%%%%%%
\section{Resolvent integral and its block encoding}
\label{sec:riesz}
A common framework for developing and analyzing quantum algorithms in linear algebra is the \emph{block encoding} framework~\cite{gilyen2019}, which provides a unitary model of access to an input matrix $A$. To describe it, consider a typical quantum computation consisting of three steps: we first initialize an ancilla register in a fixed state $\ket{0}$, then apply a joint unitary $O_A$ to the ancilla and input register, and finally measure the ancilla with the desired measurement outcome flagged by $\ket{0}$. Conditioned on this outcome, the computation enacts a transformation on the input register. Bundling these steps together, we may abstract away the ancilla and refocus on how the overall procedure transforms the input state, which is precisely the action of the matrix $A$ we wish to access.

\begin{figure}[t]
\centering
\qquad\qquad\qquad\qquad\includegraphics{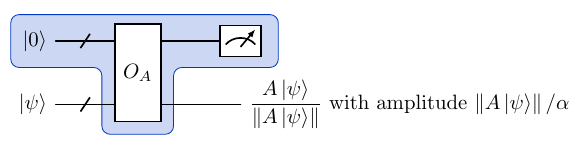}
\caption{Block encoding as a unitary access model: initialize the ancilla register in $\ket{0}$, apply the joint unitary $O_A$, and measure the ancilla. Suppose that $A/\alpha = \left( \bra{0} \otimes I \right) O_A \left( \ket{0} \otimes I \right)$. Conditioned on the desired outcome $\ket{0}$, the input register $\ket{\psi}$ is transformed by the matrix $A$ with success amplitude $\norm{A\ket{\psi}}/\alpha$.}
\label{fig:block-encoding}
\end{figure}

Formally, a unitary $O_A$ is said to be a block encoding of a matrix $A$ if
\begin{equation}
  A = \left( \bra{0} \otimes I \right) O_A \left( \ket{0} \otimes I \right),
  \label{eq:block-encoding-def}
\end{equation}
where $\ket{0}$ denotes the fixed initial state of the ancilla register. Equivalently, using the matrix representation $\ket{0}=\left[\begin{smallmatrix}
    1\\
    0
\end{smallmatrix}\right]$ and $\bra{0}=\left[\begin{smallmatrix}
    1 & 0
\end{smallmatrix}\right]$, $A$ appears as the top-left block of $O_A$,
\begin{equation}
  O_A =
  \begin{bmatrix}
    A & \cdot \\
    \cdot & \cdot
  \end{bmatrix},
  \label{eq:block-encoding-matrix}
\end{equation}
with the remaining blocks left unspecified, subject only to the unitarity of $O_A$. A key feature of this framework is that common matrix arithmetic operations can be performed directly at the level of block encodings: given block encodings of matrices $A$ and $B$, one can construct block encodings of their linear combination, product $AB$, Hermitian conjugate $A^\dagger$, and tensor product $A \otimes B$, as well as the inverse $A^{-1}$~\cite{gilyen2019}. More generally, the quantum singular value transformation allows one to block encode a polynomial transformation of the singular values of $A$, with query complexity linear in the polynomial degree.

For a general matrix $A$, the above definition must be relaxed to allow for a normalization factor $\alpha$, so that the unitary $O_A$ block encodes the rescaled matrix $A/\alpha$ rather than $A$ itself:
\begin{equation}
  A/\alpha = \left( \bra{0} \otimes I \right) O_A \left( \ket{0} \otimes I \right).
  \label{eq:block-encoding-normalized}
\end{equation}
Since the left-hand side is a submatrix of a unitary, its spectral norm is at most one, and hence $\alpha \ge \norm{A}$ is necessary for such a block encoding to exist. This lower bound is purely a property of $A$; in practice, the normalization factor achieved by a specific construction depends on the algorithm used and can be larger than $\norm{A}$. The normalization factor directly governs the cost of the resulting quantum algorithm. Indeed, when $O_A$ is applied to an input state $\ket{\psi}$, the desired outcome is heralded with success amplitude $\norm{A\ket{\psi}}/\alpha$, so a larger $\alpha$ means a smaller success probability and a correspondingly higher cost to prepare the output near deterministically. Reducing the normalization factor of a block encoding is therefore essential to reducing the overall cost of quantum algorithms. See \fig{block-encoding} for a circuit illustration of this point.

We now introduce the central problem of this work. Given an arbitrary square matrix $A$, the underlying Hilbert space decomposes into the generalized eigensubspaces of $A$, associated with its complex eigenvalues $\operatorname{\mathbf{Spec}}(A)$. Any partition of the spectrum into two disjoint groups $\operatorname{\mathbf{Spec}}_{\mathrm{in}}(A)$ and $\operatorname{\mathbf{Spec}}_{\mathrm{out}}(A)$ then induces a corresponding decomposition of the space into a pair of complementary invariant subspaces, together with an associated eigenprojection $\Pi_{\mathrm{in}}$, $\Pi_{\mathrm{out}}$ onto each of them. For presentational purpose, we first consider the case in which the eigenvalues are separated by the unit circle, so that the spectrum splits into those lying inside the circle and those lying outside. The problem is then to perform the oblique eigenprojection $\Pi_{\mathrm{in}}$ onto the generalized eigensubspaces whose eigenvalues $\operatorname{\mathbf{Spec}}_{\mathrm{in}}(A)$ lie inside the unit circle.

The eigenprojection $\Pi_{\mathrm{in}}$ admits a classical representation as a contour integral of the resolvent $(\zeta I - A)^{-1}$, known as the Riesz projector~\cite{humpherys2017}. Taking the target region to be the unit disk, the contour is the unit circle, which we parametrize as $\zeta = \mathrm{e}^{\mathrm{i}\theta}$ with $\mathrm{d}\zeta = \mathrm{i}\,\mathrm{e}^{\mathrm{i}\theta}\,\mathrm{d}\theta$. Factoring $\mathrm{e}^{\mathrm{i}\theta}$ out of the resolvent then yields the simplified expression
\begin{equation}
  \Pi_{\mathrm{in}}
  = \frac{1}{2\pi\mathrm{i}} \int_{\abs{\zeta}=1} (\zeta I - A)^{-1}\, \mathrm{d}\zeta
  = \frac{1}{2\pi} \int_0^{2\pi} \left( I - \mathrm{e}^{-\mathrm{i}\theta} A \right)^{-1} \mathrm{d}\theta,
  \label{eq:riesz-circle}
\end{equation}
where the integrand depends on $A$ only through the shifted resolvent $(I - \mathrm{e}^{-\mathrm{i}\theta} A)^{-1}$ along the unit circle.

We now review the standard approach to implementing this oblique eigenprojection on a quantum computer~\cite{RZL2026}. We first discretize the integral  using the $n$-point trapezoidal rule at the equally spaced nodes $\theta_j = 2\pi j/n$, giving $\Pi_{\mathrm{in}} \approx \frac{1}{n} \sum_{j=0}^{n-1} \left( I - \omega_n^{-j} A \right)^{-1}$, $\omega_n = \mathrm{e}^{2\pi\mathrm{i}/n}$, which converges rapidly as $n$ increases. To implement the right-hand side, we assemble the block-diagonal matrix $M = \sum_{j=0}^{n-1} \ketbra{j}{j}\otimes\left( I - \omega_n^{-j} A \right)$, each of whose blocks is a shifted copy of $A$ and hence admits a block encoding through linear combination. Applying QSVT to invert $M$ yields a block encoding of
\begin{equation}
M^{-1}
\;=\;
\begin{bmatrix}
  \left( I - A \right)^{-1} & & & \\[4pt]
   & \left( I - \omega_n^{-1} A \right)^{-1} & & \\[4pt]
   & & \ddots & \\[4pt]
   & & & \left( I - \omega_n^{-(n-1)} A \right)^{-1}
\end{bmatrix}.
\end{equation}
Preparing the uniform superposition $\operatorname{\mathbf{Had}} \ket{0} = \tfrac{1}{\sqrt{n}} \sum_{j=0}^{n-1} \ket{j}$ by a Hadamard operation on each side and extracting the average of the diagonal blocks, we obtain
\begin{equation}
  \left( \bra{0}\operatorname{\mathbf{Had}} \otimes I \right) M^{-1} \left( \operatorname{\mathbf{Had}}\ket{0} \otimes I \right)
  = \frac{1}{n} \sum_{j=0}^{n-1} \left( I - \omega_n^{-j} A \right)^{-1}
  \approx \Pi_{\mathrm{in}},
  \label{eq:hadamard-sandwich}
\end{equation}
so that the composite operation block encodes the Riesz projector. Equivalently, in the matrix representation, the Hadamard-conjugated inverse is an $n \times n$ array of blocks,
\begin{equation}
  \left( \operatorname{\mathbf{Had}} \otimes I \right) M^{-1} \left( \operatorname{\mathbf{Had}} \otimes I \right)
  \approx \begin{bmatrix}
    \Pi_{\mathrm{in}} & \cdot & \cdots & \cdot \\
    \cdot & \cdot & \cdots & \cdot \\
    \vdots & \vdots & \ddots & \vdots \\
    \cdot & \cdot & \cdots & \cdot
  \end{bmatrix},
  \label{eq:corner}
\end{equation}
with the eigenprojection $\Pi_{\mathrm{in}}$ appearing in the top-left block and all remaining blocks left unspecified.

Let us examine the normalization factor incurred by this approach. Using QSVT, each shifted resolvent can be block encoded with normalization factor $\operatorname{\pmb{\Theta}}\!\left( \norm{(I - \omega_n^{-j} A)^{-1}} \right)$ to accuracy $\epsilon$, using $\operatorname{\pmb{\Theta}}\!\left( \norm{(I - \omega_n^{-j} A)^{-1}} \log(1/\epsilon) \right)$ queries to the input block encoding of $A$. Taking the linear combination then yields a block encoding of $\Pi_{\mathrm{in}}$ with normalization factor $\operatorname{\pmb{\Theta}}(\alpha_{\mathrm{res}})$ at the same query complexity, where 
\begin{equation}
    \alpha_{\mathrm{res}}=\max_{\theta \in [0, 2\pi)} \norm{\left(I - \mathrm{e}^{-\mathrm{i}\theta} A\right)^{-1}}.
\end{equation}
This can be far larger than the desired value $\norm{\Pi_{\mathrm{in}}}$: already when $A$ is normal, $\alpha_{\mathrm{res}}=1/\delta$ scales as the inverse distance between the spectrum of $A$ and the unit circle, whereas $\norm{\Pi_{\mathrm{in}}} = 1$. See~\append{resolvent} for details. As noted earlier, this inflated normalization factor increases the cost of the resulting quantum algorithm, especially when the eigenprojection is used as a subroutine within a larger computation, such as in solving the Riccati equations discussed in this paper.

%%%%%%%%%%%%%%%%%%%%%%%%%%%%%%%%%%%%%%%%%%%%%%%%%%%%%%%%%%%%%%%%%%%%%%%%%%%%%%
\section{Two-sided block preconditioning of resolvent integral}
\label{sec:precond}
The inefficiency of the standard approach can be traced to how little it uses of the matrix $M^{-1}$. As shown above, the Hadamard conjugation exposes only a single block, the eigenprojection $\Pi_{\mathrm{in}}$, while every remaining block is left undetermined by the construction. A valid block encoding must nonetheless accommodate all of these blocks, and normalizing the entire matrix forces the large factor $\alpha_{\mathrm{res}}$, even though only the top-left block is the oblique eigenprojection of interest.

Our key insight is that these blocks are not, in fact, arbitrary. Rather than conjugating $M^{-1}$ by the Hadamard operation, which extracts only the average of the diagonal blocks, we conjugate by the quantum Fourier transform. This resolves $M^{-1}$ into a complete
array of blocks,
\begin{equation}
  \left( \operatorname{\mathbf{QFT}} \otimes I \right) M^{-1} \left( \operatorname{\mathbf{QFT}}^\dagger \otimes I \right)
  \;=\;
\begin{bmatrix}
  S_0     & S_{n-1}  & S_{n-2}  & \cdots & S_1      \\
  S_1     & S_0      & S_{n-1}  & \cdots & S_2      \\
  S_2     & S_1      & S_0      & \cdots & S_3      \\
  \vdots  & \vdots   & \vdots   & \ddots & \vdots   \\
  S_{n-1} & S_{n-2}  & S_{n-3}  & \cdots & S_0
\end{bmatrix}
  \label{eq:qft-identity}
\end{equation}
in which every block is explicitly determined: the $(k,l)$ block depends only on the difference $k - l$ and is given by the discrete Fourier transform of the shifted resolvent 
\begin{equation}
  S_m \;=\; \frac{1}{n} \sum_{j=0}^{n-1} \omega_n^{jm} \left( I - \omega_n^{-j} A \right)^{-1}.
  \label{eq:fourier-mode}
\end{equation}
In particular, the top-left block $S_0=\frac{1}{n} \sum_{j=0}^{n-1} \left( I - \omega_n^{-j} A \right)^{-1}$ approximately recovers the eigenprojection $\Pi_{\mathrm{in}}$, exactly as in the Hadamard construction, but now the surrounding blocks carry structure we can exploit rather than merely accommodate.

To make this structure explicit, we split $A$ across the two complementary invariant subspaces. Writing $\Pi_{\mathrm{out}} = I - \Pi_{\mathrm{in}}$ for the complementary eigenprojection, we define
\begin{equation}
  A_{\mathrm{in}} \;=\; A\,\Pi_{\mathrm{in}}, \qquad A_{\mathrm{out}} \;=\; A\,\Pi_{\mathrm{out}},
  \label{eq:compressions}
\end{equation}
so that $A = A_{\mathrm{in}} + A_{\mathrm{out}}$, with $A_{\mathrm{in}}$ supported on the image of eigenprojection $\mathcal{H}_{\mathrm{in}} = \operatorname{\mathbf{Im}}(\Pi_{\mathrm{in}})$ and $A_{\mathrm{out}}$ on $\mathcal{H}_{\mathrm{out}} = \operatorname{\mathbf{Im}}(\Pi_{\mathrm{out}})$. By construction, $A_{\mathrm{in}}$ has all its eigenvalues inside the unit circle, hence spectral radius less than one. The restriction of $A$ to $\mathcal{H}_{\mathrm{out}}$ has all its eigenvalues outside the unit circle and is thus invertible on that subspace; we write $A_{\mathrm{out}}^{-1}$ for the inverse of this restriction, so that $A_{\mathrm{out}}^{-1}$ has spectral radius less than one on $\mathcal{H}_{\mathrm{out}}$, extended by zero on $\mathcal{H}_{\mathrm{in}}$.

We can now evaluate the discrete Fourier transform in closed form. On each subspace, the shifted resolvent expands as a geometric series: in nonnegative powers of $A_{\mathrm{in}}$ on $\mathcal{H}_{\mathrm{in}}$, and in negative powers of $A_{\mathrm{out}}$ on $\mathcal{H}_{\mathrm{out}}$. Carrying out the discrete Fourier sum then yields, for $0 \le m \le n-1$,
\begin{equation}
  S_m \;=\;
  \begin{cases}
    A_{\mathrm{in}}^{\,m}\left(I - A_{\mathrm{in}}^{\,n}\right)^{-1}
      & \text{on } \mathcal{H}_{\mathrm{in}}, \\[4pt]
    -\,A_{\mathrm{out}}^{\,m-n}\left(I - A_{\mathrm{out}}^{-n}\right)^{-1}
      & \text{on } \mathcal{H}_{\mathrm{out}}.
  \end{cases}
  \label{eq:mode-closed-form}
\end{equation}
Since $\Pi_{\mathrm{in}}$ acts as the identity on $\mathcal{H}_{\mathrm{in}}$ and as zero on $\mathcal{H}_{\mathrm{out}}$, the deviation of the zeroth mode from the oblique eigenprojection is
\begin{equation}
  S_0 - \Pi_{\mathrm{in}} \;=\;
  \begin{cases}
    A_{\mathrm{in}}^{\,n}\left(I - A_{\mathrm{in}}^{\,n}\right)^{-1}
      & \text{on } \mathcal{H}_{\mathrm{in}}, \\[4pt]
    -\,A_{\mathrm{out}}^{-n}\left(I - A_{\mathrm{out}}^{-n}\right)^{-1}
      & \text{on } \mathcal{H}_{\mathrm{out}}.
  \end{cases}
  \label{eq:mode-zero-error}
\end{equation}
The correction terms $A_{\mathrm{in}}^{\,n}$ and $A_{\mathrm{out}}^{-n}$ are exponentially small in $n$, since both $A_{\mathrm{in}}$ and $A_{\mathrm{out}}^{-1}$ have spectral radius below one. Thus $S_0 = \Pi_{\mathrm{in}} + \operatorname{\mathbf{O}}(\norm{A_{\mathrm{in}}^{\,n}}\norm{\Pi_{\mathrm{in}}} + \norm{A_{\mathrm{out}}^{-n}}\norm{\Pi_{\mathrm{out}}})$ recovering $S_0 \approx \Pi_{\mathrm{in}}$, while the remaining are uniformly bounded as $S_m = \operatorname{\mathbf{O}}\!\left(\alpha_{\mathrm{in}}\norm{\Pi_{\mathrm{in}}} + \alpha_{\mathrm{out}}\norm{\Pi_{\mathrm{out}}}\right)$, where
$\alpha_{\mathrm{in}} = \max_{0 \le j \le n-1} \norm{A_{\mathrm{in}}^{\,j}}$ and
$\alpha_{\mathrm{out}} = \max_{1 \le k \le n} \norm{A_{\mathrm{out}}^{-k}}$ collect the
largest powers of $A_{\mathrm{in}}$ and $A_{\mathrm{out}}^{-1}$. We formally state and prove this discrete Fourier transform formula for matrix resolvent in~\thm{app-dft-resolvent}.

The block preconditioning technique we now describe was introduced in recent work on quantum linear system solvers, where it was employed to obtain an optimal query complexity of block encoding and improved complexity of initial state preparation~\cite{Low2026quantumlinearsystem}, and subsequently to develop solvers whose cost is governed by instance-dependent quantities beyond the worst-case condition number~\cite{li2025new, DLS2026}. At a high level, the idea is to block encode a \emph{scaling operator} and bundle it together with the input matrix, so that the composite block encoding represents a rescaled version of the original. A judicious choice of scaling operator leaves the condition number of the matrix unchanged, while amplifying selected blocks---equivalently, reducing their normalization factors. It is precisely this selective reduction of the normalization factor that we exploit below: when applied to the Fourier modes of the resolvent, it drives the normalization of the $\Pi_{\mathrm{in}}$ block encoding down from $\alpha_{\mathrm{res}}$ toward its optimal value $\norm{\Pi_{\mathrm{in}}}$.

We use a two-sided variant of the block preconditioning. Let $T_1$ and $T_2$ be scaling operators, diagonal in the computational basis, that act nontrivially only on the top-left entry,
\begin{equation}
  T_1 =
  \begin{bmatrix}
    t_1 & & & \\
        & 1 & & \\
        & & \ddots & \\
        & & & 1
  \end{bmatrix},
  \qquad
  T_2 =
  \begin{bmatrix}
    t_2 & & & \\
        & 1 & & \\
        & & \ddots & \\
        & & & 1
  \end{bmatrix},
  \qquad 0 < t_1, t_2 < 1.
  \label{eq:scaling-operators}
\end{equation}
Rather than inverting $M$ directly, we assemble a block encoding of the preconditioned matrix $\left( T_1 \otimes I \right)
  \left( \operatorname{\mathbf{QFT}} \otimes I \right)
  M
  \left( \operatorname{\mathbf{QFT}}^\dagger \otimes I \right)
  \left( T_2 \otimes I \right)$,
which is available at no additional query cost, since the quantum Fourier transform $\operatorname{\mathbf{QFT}}$ and the diagonal scalings $T_1, T_2$ can all be block encoded with normalization factor one. Applying QSVT to invert it then yields a block encoding of
\begin{equation}
  \left( T_2^{-1} \otimes I \right)
  \left( \operatorname{\mathbf{QFT}} \otimes I \right)
  M^{-1}
  \left( \operatorname{\mathbf{QFT}}^\dagger \otimes I \right)
  \left( T_1^{-1} \otimes I \right)
  \;=\;
\begin{bmatrix}
  \dfrac{1}{t_1 t_2}\, S_0 & \dfrac{1}{t_2}\, S_{n-1} & \dfrac{1}{t_2}\, S_{n-2} & \cdots & \dfrac{1}{t_2}\, S_1 \\[10pt]
  \dfrac{1}{t_1}\, S_1     & S_0                      & S_{n-1}                  & \cdots & S_2 \\[10pt]
  \dfrac{1}{t_1}\, S_2     & S_1                      & S_0                      & \cdots & S_3 \\[10pt]
  \vdots                   & \vdots                   & \vdots                   & \ddots & \vdots \\[10pt]
  \dfrac{1}{t_1}\, S_{n-1} & S_{n-2}                  & S_{n-3}                  & \cdots & S_0
\end{bmatrix}.
  \label{eq:preconditioned-inverse}
\end{equation}
This implies
\begin{equation}
\begin{aligned}
\alpha_{\mathrm{precond}}&=\norm{\left( T_2^{-1} \otimes I \right)
  \left( \operatorname{\mathbf{QFT}} \otimes I \right)
  M^{-1}
  \left( \operatorname{\mathbf{QFT}}^\dagger \otimes I \right)
  \left( T_1^{-1} \otimes I \right)} \\
&\le\;
\frac{1}{t_1 t_2}\,\norm{S_0}
\;+\; \frac{\sqrt{n-1}}{t_2}\,\max_{m}\norm{S_m}
\;+\; \frac{\sqrt{n-1}}{t_1}\,\max_{m}\norm{S_m}
\;+\; \norm{M^{-1}}\\
&=\operatorname{\mathbf{O}}\;
\left(\frac{\norm{\Pi_{\mathrm{in}}}}{t_1 t_2}
    \;+\; \left( \frac{\sqrt{n}}{t_1} + \frac{\sqrt{n}}{t_2} \right)
      \left( \alpha_{\mathrm{in}}\norm{\Pi_{\mathrm{in}}} + \alpha_{\mathrm{out}}\norm{\Pi_{\mathrm{out}}} \right)
    \;+\; \alpha_{\mathrm{res}}\right).
\end{aligned}
\label{eq:precond-norm-bound}
\end{equation}
Here, the four terms bound the four structural parts of the inverse matrix: the top-left entry $\tfrac{1}{t_1 t_2} S_0$, the remainder of the top row (scaled by $\tfrac{1}{t_2}$), the remainder of the left column (scaled by $\tfrac{1}{t_1}$), and the bottom-right $(n-1) \times (n-1)$ block of unscaled modes. Continuing, the bottom-right block has norm at most $\norm{M^{-1}} = \max_{0 \le j \le n-1} \norm{(I - \omega_n^{-j} A)^{-1}} \le \alpha_{\mathrm{res}}$ due to the resolvent bound. The top row and left column are controlled by the individual modes, which satisfy $\max_{m} \norm{S_m} = \operatorname{\mathbf{O}}(\alpha_{\mathrm{in}}\norm{\Pi_{\mathrm{in}}} + \alpha_{\mathrm{out}}\norm{\Pi_{\mathrm{out}}})$, while the top-left entry is exponentially close to $\norm{\Pi_{\mathrm{in}}}$. 

Quantities such as $\alpha_{\mathrm{in}}, \alpha_{\mathrm{out}}, \alpha_{\mathrm{res}}$ can all be controlled under a natural geometric assumption on the spectrum. Recall that the \emph{numerical range} of a matrix $B$ is the set
\begin{equation}
  \mathcal{W}(B) \;=\; \left\{\, \bra{\psi} B \ket{\psi} \;:\; \norm{\ket{\psi}} = 1 \,\right\},
  \label{eq:numerical-range}
\end{equation}
a compact convex subset of the complex plane. It may be regarded as a relaxed notion of the spectrum: it always contains the eigenvalues of $B$, and coincides with their convex hull when $B$ is normal, but is generally larger for a non-normal matrix. The numerical range is a standard tool for studying non-normal matrices~\cite{TrefethenEmbree05}, and has been widely used in previous quantum algorithms for non-normal inputs~\cite{QEVP, niying2026, wang2026lchm,Krovi2023improvedquantum}. Its usefulness here stems from the Crouzeix--Palencia bound~\cite{CrouzeixPalencia17}: for any function $f$ analytic on a neighborhood of $\mathcal{W}(B)$,
\begin{equation}
  \norm{f(B)} \;\le\; (1 + \sqrt{2}) \max_{z \in \mathcal{W}(B)} \abs{f(z)},
  \label{eq:crouzeix-palencia}
\end{equation}
which bounds the norm of a matrix function purely by the behavior of the scalar function $f$ over $\mathcal{W}(B)$. Very recently, it has been claimed that the above constant $1+\sqrt{2}$ can be improved to the optimal value $2$, resolving Crouzeix's conjecture~\cite{lorist2026crouzeix, numrange2spectral}. Using this improved constant would reduce the constant prefactor of our algorithm, but leave the asymptotic analysis unchanged.

Let us analyze the block preconditioning under the numerical-range assumption. Suppose the per-block numerical ranges are separated from the unit circle by a gap $\delta$,
\begin{equation}
  \mathcal{W}(A_{\mathrm{in}}) \subseteq \{\,z \colon |z| \le 1 - \delta \,\}, \qquad
  \mathcal{W}(A_{\mathrm{out}}^{-1}) \subseteq \{\, z \colon |z| \le 1 - \delta \,\}.
  \label{eq:nr-assumption}
\end{equation}
By the Crouzeix--Palencia bound applied on each subspace~\cite{CrouzeixPalencia17}, the powers are uniformly bounded: $\norm{A_{\mathrm{in}}^{\,j}} \le (1+\sqrt{2})(1-\delta)^{j}$, $  \norm{A_{\mathrm{out}}^{-k}} \le (1+\sqrt{2})(1-\delta)^{k}$, so that $\alpha_{\mathrm{in}}, \alpha_{\mathrm{out}} = \operatorname{\mathbf{O}}(1)$. Summing the geometric series and applying~\cite{CrouzeixPalencia17}, we bound the resolvent likewise:
\begin{equation}
  \alpha_{\mathrm{res}} \;=\; \operatorname{\mathbf{O}}\!\left( \frac{\norm{\Pi_{\mathrm{in}}}}{\delta} \right),
  \label{eq:alpha-res-nr}
\end{equation}
Finally, applying the Crouzeix--Palencia bound again, the discretization error $\norm{S_0 - \Pi_{\mathrm{in}}} = \operatorname{\mathbf{O}}((1-\delta)^{n}\norm{\Pi_{\mathrm{in}}})$ falls below $\epsilon$ as long as $n = \operatorname{\mathbf{O}}\!\left( \tfrac{1}{\delta}\log\tfrac{\norm{\Pi_{\mathrm{in}}}}{\epsilon} \right)$.
Substituting $\alpha_{\mathrm{in}}, \alpha_{\mathrm{out}} = \operatorname{\mathbf{O}}(1)$, $\alpha_{\mathrm{res}} = \operatorname{\mathbf{O}}(\norm{\Pi_{\mathrm{in}}}/\delta)$ and $n = \operatorname{\mathbf{O}}\!\left( \tfrac{1}{\delta}\log\tfrac{\norm{\Pi_{\mathrm{in}}}}{\epsilon} \right)$ with the balanced choice $t_1 = t_2 
\;=\;\operatorname{\pmb{\Theta}}\!\left( \sqrt{\delta} \right)$, the preconditioned inverse has norm at most 
\begin{equation}
  \alpha_{\mathrm{precond}}
  \;=\; \operatorname{\mathbf{O}}\!\left( \frac{\norm{\Pi_{\mathrm{in}}}}{\delta}
  \sqrt{\log\tfrac{\norm{\Pi_{\mathrm{in}}}}{\epsilon}} \right),
  \label{eq:alpha-precond}
\end{equation}
which sets the condition number for the QSVT inversion. Meanwhile, extracting the top-left entry and rescaling, the resulting block encoding of
$\Pi_{\mathrm{in}}$ has reduced normalization factor
\begin{equation}
  \alpha_{\Pi} 
  \;=\; t_1t_2\alpha_{\mathrm{precond}}
  \;=\; \operatorname{\mathbf{O}}\!\left( \norm{\Pi_{\mathrm{in}}} \sqrt{\log\tfrac{\norm{\Pi_{\mathrm{in}}}}{\epsilon}} \right).
  \label{eq:alpha-Pi-nr}
\end{equation}
Compared with the unpreconditioned algorithm, whose condition number is $\operatorname{\mathbf{O}}(\alpha_{\mathrm{res}}) = \operatorname{\mathbf{O}}(\norm{\Pi_{\mathrm{in}}}/\delta)$, the preconditioned inversion has condition number $\alpha_{\mathrm{precond}}=\operatorname{\mathbf{O}}\!\left( \tfrac{\norm{\Pi_{\mathrm{in}}}}{\delta}\sqrt{\log\tfrac{\norm{\Pi_{\mathrm{in}}}}{\epsilon}} \right)$, larger only by a factor $\sqrt{\log\tfrac{\norm{\Pi_{\mathrm{in}}}}{\epsilon}}$. The query complexity of the eigenprojection therefore remains almost the same as the standard approach, while the normalization factor is significantly reduced to nearly its optimal value $\norm{\Pi_{\mathrm{in}}}$. 

The numerical range enters only through the Crouzeix--Palencia bound, which makes it a specific instance of \emph{spectral set}. We say a compact set $\mathcal{K}$ is a spectral set with constant $c$ if it contains the spectrum of $B$ such that
\begin{equation}
    \norm{g(B)} \leq c \max_{z \in \mathcal{K}} \abs{g(z)}
\end{equation}
for every $g$ analytic near $\mathcal{K}$. For a normal matrix $B$, the spectrum is a spectral set with $c = 1$. For non-normal $B$, one may enlarge the set $\mathcal{K}$ or pay a larger constant $c$. Taking $g(z) = z^{j}$ on a spectral set inside $\{z\colon \abs{z} \leq 1-\delta\}$ gives the geometric decay $\norm{B^{j}} \leq c(1-\delta)^{j}$ that our analysis needs. So any such spectral sets for $A_{\mathrm{in}}$ and $A_{\mathrm{out}}^{-1}$ suffice, with $c$ entering the normalization factor linearly. For example, when $A$ is diagonalizable, the spectrum is a spectral set with $c$ equal to the eigenbasis condition number, and no numerical-range assumption is needed. We defer this general treatment to~\append{precond}, with the quantum oblique eigenprojection algorithm established in~\thm{app-disk-main} for the unit disk.

%%%%%%%%%%%%%%%%%%%%%%%%%%%%%%%%%%%%%%%%%%%%%%%%%%%%%%%%%%%%%%%%%%%%%%%%%%%%%%
\section{Quantum oblique eigenprojection for more general regions}
\label{sec:regions}
In the previous section, we described a quantum algorithm for the oblique eigenprojection onto the generalized eigensubspaces of $A$ associated with eigenvalues inside the unit disk. Our approach is based on a two-sided block preconditioning of the matrix resolvent, which is easiest to explain for the unit disk, where the discrete Fourier transform of the resolvent admits a closed-form expression. However, the same approach can be adapted to more general regions of the complex plane, allowing us to project onto the eigensubspace associated with any region whose boundary separates the spectrum of $A$ into $\operatorname{\mathbf{Spec}}_{\mathrm{in}}(A)$ and $\operatorname{\mathbf{Spec}}_{\mathrm{out}}(A)$. We describe this extension in the present section.

Let $\mathcal{C}$ be a closed curve enclosing $\operatorname{\mathbf{Spec}}_{\mathrm{in}}(A)$ but not $\operatorname{\mathbf{Spec}}_{\mathrm{out}}(A)$, and disjoint from $\operatorname{\mathbf{Spec}}(A)$. The oblique eigenprojection onto $\operatorname{\mathbf{Spec}}_{\mathrm{in}}(A)$ is the Riesz projector $\Pi_{\mathrm{in}} = \frac{1}{2\pi\mathrm{i}} \int_{\mathcal{C}} (w I - A)^{-1}\, \mathrm{d}w$. Let $\mathcal{S}_{\mathrm{in}}$ and $\mathcal{S}_{\mathrm{out}}$ be subsets of the complex plane satisfying: 
\begin{enumerate}
  \item[(i)] $\operatorname{\mathbf{Spec}}_{\mathrm{in}}(A) \subseteq \mathcal{S}_{\mathrm{in}}$ and $\operatorname{\mathbf{Spec}}_{\mathrm{out}}(A) \subseteq \mathcal{S}_{\mathrm{out}}$;
  \item[(ii)] $\mathcal{S}_{\mathrm{in}} \cap \mathcal{S}_{\mathrm{out}} = \varnothing$;
  \item[(iii)] $\mathcal{S}_{\mathrm{in}}$ and $\mathcal{S}_{\mathrm{out}}$ are open.
\end{enumerate}
Now let $f$ be analytic on $\mathcal{S}_{\mathrm{in}} \cup \mathcal{S}_{\mathrm{out}}$ and separate the two components in modulus,
\begin{equation}
  |f(w)| < 1 \ \ (w \in \mathcal{S}_{\mathrm{in}}), \qquad
  |f(w)| > 1 \ \ (w \in \mathcal{S}_{\mathrm{out}}).
  \label{eq:separation}
\end{equation}
By the holomorphic functional calculus, the eigenprojection of $A$ for the general region then coincides with the eigenprojection of $f(A)$ for the unit disk:
\begin{equation}
  \Pi_{\mathrm{in}}
  \;=\; \frac{1}{2\pi\mathrm{i}} \int_{\mathcal{C}} (w I - A)^{-1}\, \mathrm{d}w
  \;=\; \frac{1}{2\pi\mathrm{i}} \int_{|z|=1} \bigl(z I - f(A)\bigr)^{-1}\, \mathrm{d}z.
  \label{eq:cov}
\end{equation}
One can view this as a change-of-variable formula for resolvent integration, though it is an identity of projectors rather than a literal substitution. Specifically, both integrals compute the Riesz projector, the left in the spectral plane of $A$ and the right in that of $f(A)$, and these agree because $A$ and $f(A)$ share the same invariant subspaces while $f$ carries $\operatorname{\mathbf{Spec}}_{\mathrm{in}}(A)$ inside the unit circle and $\operatorname{\mathbf{Spec}}_{\mathrm{out}}(A)$ outside by the spectral mapping theorem~\cite[Theorem 12.7.1]{humpherys2017}\cite[Theorem 8.3]{roman2008advanced}. The general-region problem thus reduces to the unit-disk construction applied to $f(A)$ in place of $A$. 

Crucially, $f$ only needs to separate $\mathcal{S}_{\mathrm{in}}$ and $\mathcal{S}_{\mathrm{out}}$ by modulus---it need not be injective, conformal, or boundary-preserving---so we are free to choose it to be a low-degree rational function, making $f(A)$ inexpensive to block encode. This suggests separating the design of the algorithm from the construction of $f$: we first treat $f(A)$ as an oracle, assuming a block encoding of it is given, and apply the oblique eigenprojection to the unit disk. In this oracle model, the construction of \sec{precond} applies verbatim with $A$ replaced by $f(A)$, yielding a block encoding of $\Pi_{\mathrm{in}}$ with normalization factor close to $\norm{\Pi_{\mathrm{in}}}$ and query complexity governed by the gap of $f(A)$ to the unit circle. The oracle model requires only a relatively \emph{weak} analyticity hypothesis. Because $f(A)$ is given as input rather than approximated, $f$ need only be analytic on neighborhoods of the two block numerical ranges $\mathcal{S}_{\mathrm{in}}=\mathcal{W}(A_{\mathrm{in}})$ and $\mathcal{S}_{\mathrm{out}}=\mathcal{W}(A_{\mathrm{out}})$ \emph{separately}, and not on the full numerical range $\mathcal{W}(A)$. This distinction is important: $\mathcal{W}(A)$ is convex and generally strictly larger than the two block ranges, so a singularity of $f$ may lie inside $\mathcal{W}(A)$ while remaining outside of each block range. This makes the oracle model more flexible, as it admits any separating map whose singularities fall between the blocks, including maps that an explicit polynomial approximation of $f(A)$ over the whole numerical range could not accommodate. We formally present this algorithm in~\thm{app-region-main}.

We now turn to constructing a block encoding of $f(A)$ from a block encoding of $A$. Let $\mathcal{K}$ be a spectral set for $A$, meaning that $\norm{g(A)}\le c \sup_{z \in \mathcal{K}} |g(z)|$ for every function $g$ analytic on a neighborhood of $\mathcal{K}$, where $c$ is a constant depending only on $\mathcal{K}$ and $A$. Let $f$ be analytic on $\mathcal{K}$ and separate the spectrum,
\begin{equation}
  |f(w)| < 1 \ \ (w \in \operatorname{\mathbf{Spec}}_{\mathrm{in}}(A)), \qquad
  |f(w)| > 1 \ \ (w \in \operatorname{\mathbf{Spec}}_{\mathrm{out}}(A)).
  \label{eq:separation-spectrum}
\end{equation}
Approximating $f$ by a polynomial $p$ and applying the definition of spectral set to $g = f - p$, we obtain the error bound $\norm{f(A) - p(A)} \;\le\; c \sup_{z \in \mathcal{K}} |f(z) - p(z)|$, so a polynomial that approximates $f$ uniformly on $\mathcal{K}$ yields a block encoding of $f(A)$ to matching accuracy. As mentioned previously, a common choice is the numerical range $\mathcal{S} = \mathcal{W}(A)$, which is a spectral set for every $A$ with the dimension-independent constant $c = 1 + \sqrt{2}$ by the Crouzeix--Palencia theorem~\cite{CrouzeixPalencia17}. With this choice, $f$ need only be analytic on a neighborhood of $\mathcal{W}(A)$. The polynomial $p$ is then applied by the QEVT algorithm~\cite[Section 8 and Appendix B]{QEVP} at a cost of $\operatorname{\mathbf{O}}(d)$ queries to the block encoding of $A$, where $d$ is the degree of $p$. The required degree can be obtained from the Bernstein--Walsh theorem~\cite[Section III.3]{suetin1998}. In particular, if $f$ is analytic on a neighborhood of $\mathcal{W}(A)$, the best degree-$d$ polynomial approximation on $\mathcal{W}(A)$ decays geometrically, so that $d = \operatorname{\mathbf{O}}(\log(1/\varepsilon))$ suffices for accuracy $\varepsilon$, with the rate set by how far the nearest singularity of $f$ lies beyond $\mathcal{W}(A)$. This generalizes the approximation of a function on $[-1,1]$ through its analyticity within a Bernstein ellipse~\cite[Theorem 8.1 and 8.2]{trefethen2019approximation}. Finally, the separating map partitions the spectral set into $\mathcal{S}_{\mathrm{in}} = \{z : |f(z)| < 1\}$ and $\mathcal{S}_{\mathrm{out}} = \{z : |f(z)| > 1\}$, which are open, disjoint, and cover $\operatorname{\mathbf{Spec}}_{\mathrm{in}}(A)$ and $\operatorname{\mathbf{Spec}}_{\mathrm{out}}(A)$ respectively. Hence, the block encoding of $f(A)$ can be supplied to the unit-disk construction to produce $\Pi_{\mathrm{in}}$.

It remains to construct a separating map $f$ for a given target region. For a simply connected region, the Riemann mapping theorem guarantees a conformal map onto the unit disk, which separates the region from its complement and extends across an analytic boundary; more generally, one may relax the requirement and seek only a rational function whose modulus separates the two spectral parts, with poles placed outside the numerical range so that analyticity on $\mathcal{W}(A)$ is preserved. In practice, the regions arising in applications are bounded by lines and circular arcs---half-planes, disks, strips, and sectors---for which a separating map is built directly from elementary M\"obius factors, with no approximation required. Here, we work out the half-plane case in detail, as it underlies our applications to solving the Riccati and Sylvester equations, where the spectrum is split across a line.

To this end, suppose the spectrum is separated by the imaginary axis, with the target part $\operatorname{\mathbf{Spec}}_{\mathrm{in}}(A)$ in the open right half-plane $\Re(z) > 0$ and $\operatorname{\mathbf{Spec}}_{\mathrm{out}}(A)$ in the open left half-plane $\Re(z) < 0$. A separating map is then given by the degree-two rational function
\begin{equation}
  f(z) = \frac{(z-1)^2}{(z+1)^2},
  \label{eq:halfplane-map}
\end{equation}
with a double zero at $z = 1$ and a double pole at $z = -1$. $f$ sends the right half-plane into the unit disk, the imaginary axis onto the unit circle, and the left half-plane outside, separating the spectrum as required.

Assume for simplicity that a block encoding of $A$ is given with $\norm{A} \le \tfrac{1}{2}$. Then $A + I$ is well-conditioned: its singular values lie in $[\tfrac{1}{2}, \tfrac{3}{2}]$, since $\norm{A + I} \le 1 + \norm{A} \le \tfrac{3}{2}$ and $\norm{(A + I)^{-1}} \le (1 - \norm{A})^{-1} \le 2$, giving a condition number at most $3$. QSVT-based matrix inversion therefore block encodes $(A + I)^{-1}$ to error $\varepsilon$ with normalization $\operatorname{\mathbf{O}}(1)$ using $\operatorname{\mathbf{O}}(\log(1/\varepsilon))$ queries to the block encoding of $A$. Applying $(A + I)^{-1}$ twice and multiplying by $(A - I)^2$ yields a block encoding of $f(A)=\frac{(A-I)^2}{(A+I)^2}$ with normalization $\norm{f(A)} = \operatorname{\mathbf{O}}(1)$, again at $\operatorname{\mathbf{O}}(\log(1/\varepsilon))$ query cost. Thus, under $\norm{A} \le \tfrac{1}{2}$, the map $f$ is applied directly with only a constant-factor overhead in normalization and a logarithmic-factor query cost.

Additionally, for the reduction to be efficient, the map $f$ must not crush the gap. Suppose the spectrum is separated from the imaginary axis by $\delta$, in the sense that $|\Re(\lambda)| \ge \delta$ for every $\lambda \in \operatorname{\mathbf{Spec}}(A)$. Under the normalization $\norm{A} \le \tfrac{1}{2}$, the image spectrum $\operatorname{\mathbf{Spec}}(f(A))$ is then separated from the unit circle by a gap $\delta_f = \operatorname{\pmb{\Theta}}(\delta)$, so the map preserves the spectral gap up to a constant factor. The same holds at the level of the numerical range: if the per-block numerical ranges $\mathcal{W}(A_{\mathrm{in}})$ and $\mathcal{W}(A_{\mathrm{out}})$ are separated from the imaginary axis by $\delta$, then $\abs{f}$ is bounded by $1 - \operatorname{\pmb{\Theta}}(\delta)$ on $\mathcal{W}(A_{\mathrm{in}})$ and by $1 + \operatorname{\pmb{\Theta}}(\delta)$ from below on $\mathcal{W}(A_{\mathrm{out}})$, so applying the Crouzeix--Palencia bound to the composite function over each per-block numerical range yields the desired estimates with a comparable gap. More generally, the same argument applies to any spectral sets for $A_{\mathrm{in}}$ and $A_{\mathrm{out}}$ separated from the imaginary axis by $\delta$, with the Crouzeix--Palencia constant replaced by the spectral-set constant $c$. Consequently the corresponding cost bounds apply to $f(A)$ with $\delta$ replaced by $\delta_f = \operatorname{\pmb{\Theta}}(\delta)$. See~\fig{halfplane-separation} for an illustration. We defer a detailed treatment to~\append{region}.

\begin{figure}[t]
\centering
\includegraphics[scale=0.95]{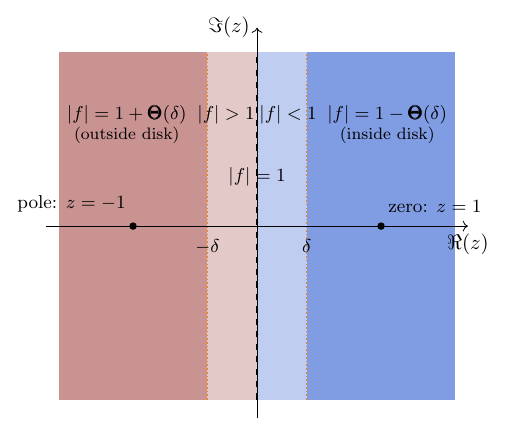}
\caption{The rational map $f(z) = (z-1)^2/(z+1)^2$ separates the complex plane by the imaginary axis. The right half-plane (blue) is mapped inside the unit circle ($|f| < 1$) and the left half-plane (red) outside ($|f| > 1$), with the imaginary axis ($|f| = 1$) mapped onto the unit circle. The double zero at $z = 1$ and double pole at $z = -1$ lie in the right and left half-planes respectively. Darker shading indicates the gapped regions $|\Re(z)| \geq \delta$, on which $|f| = 1 - \operatorname{\pmb{\Theta}}(\delta)$ (right) and $|f| = 1 + \operatorname{\pmb{\Theta}}(\delta)$ (left), so the gap to the unit circle is preserved, $\delta_f = \operatorname{\pmb{\Theta}}(\delta)$.}
\label{fig:halfplane-separation}
\end{figure}

%%%%%%%%%%%%%%%%%%%%%%%%%%%%%%%%%%%%%%%%%%%%%%%%%%%%%%%%%%%%%%%%%%%%%%%%%%%%%%
\section{Applications}
\label{sec:app}
The oblique eigenprojection developed above is a general-purpose primitive, which we now instantiate in three settings. In each, our two-sided block preconditioning reduces the normalization factor of the block-encoded projector to near its optimal value, and this reduction carries through to a lower overall cost, sharpening a recent quantum algorithm.

\subsection{Eigenstate preparation for complex spectra}
\label{sec:app-eigenstate}

Preparing an eigenstate of a given operator is among the most fundamental tasks in quantum computation, which is relevant for studying static properties of the underlying physical system. For a Hermitian Hamiltonian the eigenstates are orthogonal, and the task is well understood. Indeed, given an initial state with nonnegligible overlap on the target eigenspace and a spectral gap $\delta$ separating it from the rest of the spectrum, one can prepare the target eigenstate efficiently~\cite{lin2020}. The same task is natural for an arbitrary, non-normal matrix $A$, whose eigenvalues may lie anywhere in the complex plane and whose eigenspaces are oblique. Such matrices arise in simulating non-Hermitian physics~\cite{Ashida20}, in transcorrelated approaches to quantum chemistry~\cite{McArdle20}, and in studying irreversible Markov chains~\cite{banerjee2026}, among other settings. Here the goal is similar: given a target eigenvalue (or cluster) separated from the rest of the spectrum, and an initial state with nonnegligible overlap on the corresponding invariant subspace, one aims to prepare the associated eigenstate (or project the initial state onto the associated eigensubspaces).

When the spectrum of $A$ lies on the real line, this can be accomplished by the QEVT algorithm of Low and Su~\cite{QEVP}, which applies a polynomial transformation to the eigenvalues of $A$ through the Chebyshev and Faber history states. Choosing the polynomial to be close to one near the target eigenvalue and close to zero elsewhere, as in \fig{indicator-poly}, prepares the corresponding eigenstate at a cost scaling as $1/\delta$ in the spectral gap, recovering the nearly optimal scaling of ground-state preparation for Hermitian Hamiltonians. This filtering is effective because the target eigenvalues lie on a line: a polynomial can approximate an indicator along the real axis. However, for eigenvalues distributed in the complex plane, as depicted in \fig{target-region}, there is an intrinsic obstruction. A polynomial---indeed any analytic function---cannot approximate the indicator of a region uniformly in all directions: by the maximum modulus principle its magnitude attains its maximum on the boundary, so it cannot be close to one inside a region and close to zero just outside it~\cite{cook2020}. 
Preparing an eigenstate of a matrix with complex eigenvalues therefore calls for a genuinely different mechanism. To this end, one may instead form the Riesz projector through the resolvent integral of \sec{riesz}, as done by Rodenas-Ruiz, Zhao, and Lee~\cite{RZL2026}; this bypasses the analyticity requirement and applies to complex spectra. Its cost, however, scales as $1/\delta^2$: the block-encoded projector carries a normalization factor $\alpha_{\mathrm{res}} = \operatorname{\mathbf{O}}(\norm{\Pi_{\mathrm{in}}}/\delta)$ set by the resolvent norm on the contour, and this leads the $1/\delta$ factor to enter the eigenstate-preparation cost twice---once in the $\operatorname{\mathbf{O}}(1/\delta)$ queries needed to block encode the projector, and again in the $\operatorname{\mathbf{O}}(1/\delta)$ rounds of amplitude amplification needed to project the initial state onto the target subspace.

Our algorithm applies oblique eigenprojection at a cost matching the result of Low and Su while handling complex spectra. For simplicity, suppose the target eigenvalue or cluster is enclosed by a disk $\{z:\, |z - c| < \rho \,\}$ that separates it from the rest of the spectrum with a gap $\delta$, where $c$ and $\rho$ are here considered constants, independent of $\delta$. The separating map is the linear function $f(z) = (z - c)/\rho$, so $f(A) = (A - cI)/\rho$ is obtained from the block encoding of $A$ by a shift and a rescaling, at only a constant-factor overhead in normalization for a center $|c| = \operatorname{\mathbf{O}}(1)$. Feeding $f(A)$ into the unit-disk construction of \sec{precond} block encodes the eigenprojection $\Pi_{\mathrm{in}}$ onto the enclosed invariant subspace with a query complexity scaling as $1/\delta$ and, after two-sided block preconditioning, a normalization factor close to its optimal value $\norm{\Pi_{\mathrm{in}}}$. Because the normalization no longer grows with the inverse gap, the amplitude amplification that projects the initial state onto the target subspace needs only $\widetilde{\operatorname{\mathbf{O}}}(1)$ rounds, and the eigenstate is prepared in $\widetilde{\operatorname{\mathbf{O}}}(1/\delta)$ total queries. This quadratically improves the resolvent-integral approach and matches the scaling in $1/\delta$ of Low and Su for nonnormal matrices with real spectra---in other words, preparing an eigenstate associated with a complex eigenvalue is no more expensive, in $\delta$, than filtering a real one. See~\append{app-eigenstate} for details.

\subsection{Continuous-time algebraic Riccati equations}
\label{sec:app-care}
As a second application we consider the continuous-time algebraic Riccati equation,
\begin{equation}
  X Q X - X P - P^\dagger X - R = 0,
\end{equation}
where $P, Q, R$ are square matrices of the same size with $Q, R$ Hermitian, and the matrix $X$ is the sought solution. Riccati equations are central to optimal control, filtering, and stability theory, and in quantum chemistry they arise directly in ring coupled-cluster doubles and the random-phase approximation~\cite{RZL2026}. The stabilizing solution of the Riccati equation is characterized through an invariant subspace of the associated Hamiltonian matrix
\begin{equation}
  J = \begin{bmatrix} P & -Q \\ -R & -P^\dagger \end{bmatrix},
\end{equation}
whose spectrum is symmetric about the imaginary axis. Splitting the spectrum along the imaginary axis into a stable part ($\Re (\lambda) < 0$) and an antistable part ($\Re (\lambda) > 0$), the stable invariant subspace is spanned by the columns of $\left[\begin{smallmatrix} U \\ V \end{smallmatrix}\right]$, and the stabilizing solution is recovered as $X = V U^{-1}$. Solving the Riccati equation thus reduces to projecting onto an invariant subspace of a non-normal matrix split by the imaginary axis---an oblique eigenprojection of exactly the kind our primitive provides.

Recently, Rodenas-Ruiz, Zhao, and Lee~\cite{RZL2026} solve the Riccati equation on a quantum computer along these lines. The stable invariant subspace spanned by $\left[\begin{smallmatrix} I \\ X \end{smallmatrix}\right]$ is the kernel of the projector $\Pi_{\mathrm{in}}$ onto the antistable subspace, so writing $\Pi_{\mathrm{in}} = \left[\begin{smallmatrix} \Pi_1 & \Pi_2 \end{smallmatrix}\right]$ gives $\Pi_1 + \Pi_2 X = 0$ and hence $X = -\Pi_2^{+}\Pi_1$ where $(\cdot)^+$ denotes the pseudoinverse. They therefore block encode the antistable projector $\Pi_{\mathrm{in}}$ through the resolvent integral of \sec{riesz}, and then extract $X$ from its sub-blocks by implementing $\Pi_2^{+}$ through QSVT. Writing $\delta = \min_\lambda \abs{\Re(\lambda)}$ for the gap from the imaginary axis, the block-encoded projector carries a normalization factor $\alpha_{\mathrm{res}} = \operatorname{\mathbf{O}}(\norm{\Pi_{\mathrm{in}}}/\delta)$ set by the resolvent norm on the contour.  Beyond the $\operatorname{\mathbf{O}}(1/\delta)$ queries needed to block encode the projector, this normalization enters the extraction twice more---in the condition number of the pseudoinverse that forms $X = -\Pi_2^{+}\Pi_1$, and in the normalization of the resulting solution---so that preparing the Riccati solution costs $\widetilde{\operatorname{\mathbf{O}}}(1/\delta^3)$ queries in the gap ignoring polylogarithmic factors.

Applying our oblique eigenprojection removes the two normalization-driven factors. Since the spectrum is split by the imaginary axis, we use the half-plane construction of \sec{regions}, with a degree-two rational map sending the antistable half-plane into the unit disk and the imaginary axis onto the unit circle, and the per-block numerical-range hypotheses of \sec{precond} inherited with a transformed gap $\delta_f = \operatorname{\pmb{\Theta}}(\delta)$. Two-sided block preconditioning then block encodes the projector with a normalization factor close to its optimal value $\norm{\Pi_{\mathrm{in}}}$, in place of $\alpha_{\mathrm{res}} = \operatorname{\mathbf{O}}(\norm{\Pi_{\mathrm{in}}}/\delta)$. The pseudoinverse condition number and the solution normalization then no longer grow with the inverse gap, so the two extra factors collapse to $\widetilde{\operatorname{\mathbf{O}}}(1)$ and the Riccati solution is prepared in $\widetilde{\operatorname{\mathbf{O}}}(1/\delta)$ queries---a cubic improvement over the standard resolvent-integral solver of~\cite{RZL2026}. We describe this improvement in detail in~\append{app-riccati}.

\subsection{Ordinary Sylvester equations}
\label{sec:app-sylvester}
As a final application we consider the Sylvester equation
\begin{equation}
  A X + X B = C,
\end{equation}
where $A, B, C$ are given matrices of the same size and one seeks to construct a block-encoding of the solution $X$ with normalization within a constant factor of the optimal value $\norm{X}$~\cite{somma2025}. Sylvester equations arise throughout control theory, model reduction, and the numerical solution of matrix and differential equations. As with the Riccati equation, the solution is encoded in an invariant subspace of an associated block matrix:
\begin{equation}
  M = \begin{bmatrix} A & C \\ 0 & -B \end{bmatrix}.
\end{equation}
Its spectrum is the union of the spectra of $A$ and $-B$, and it has two invariant subspaces: the $A$ branch is spanned by $\left[\begin{smallmatrix} I \\ 0 \end{smallmatrix}\right]$, while the $-B$ branch is spanned by $\left[\begin{smallmatrix} -X \\ I \end{smallmatrix}\right]$, since the invariance condition $M \left[\begin{smallmatrix} -X \\ I \end{smallmatrix}\right] = \left[\begin{smallmatrix} -X \\ I \end{smallmatrix}\right](-B)$ is precisely the Sylvester equation $AX + XB = C$. Stacking these bases into $S = \left[\begin{smallmatrix} I & -X \\ 0 & I \end{smallmatrix}\right]$, the oblique projector onto the $A$ branch along the $B$ branch is $\Pi_{\mathrm{in}} = S \left[\begin{smallmatrix} I & 0 \\ 0 & 0 \end{smallmatrix}\right] S^{-1} = \left[\begin{smallmatrix} I & X \\ 0 & 0 \end{smallmatrix}\right]$, so the solution $X$ appears directly in its off-diagonal block. When the spectra of $A$ and $-B$ lie in opposite open half-planes---separated, after a shift and rotation, by the imaginary axis---this is again an oblique eigenprojection handled by our algorithm, now with $X$ read off from a single block of the projector. Unlike the Riccati equation, the extraction requires no further pseudoinverse.

Recently, Wang and Liu~\cite{WL2026} give a quantum algorithm for the Sylvester equation along these lines.   We compare against their direct augmented method~\cite[Remark 6.4]{WL2026}, which block encodes the projector through the resolvent integral of \sec{riesz} applied to the augmented matrix $M$, so that---exactly as for the eigenstate-preparation and Riccati applications---the block-encoded projector carries a normalization factor $\alpha_{\mathrm{res}} = \operatorname{\mathbf{O}}(\norm{\Pi_{\mathrm{in}}}/\delta)$ set by the resolvent norm on the contour with $\delta$ the gap from the imaginary axis. Since the solution is the off-diagonal block of the projector, no pseudoinverse is needed, and this normalization enters the cost only once (for amplifying the block-encoding), on top of the $\operatorname{\mathbf{O}}(1/\delta)$ queries needed to block encode the projector, so that preparing the solution costs $\widetilde{\operatorname{\mathbf{O}}}(1/\delta^2)$ queries in the gap. 

Our oblique eigenprojection applies directly and improves on this complexity. Using the half-plane construction of \sec{regions} to project onto the invariant subspace of the $A$ branch, two-sided block preconditioning replaces the normalization factor $\alpha_{\mathrm{res}} = \operatorname{\mathbf{O}}(\norm{\Pi_{\mathrm{in}}}/\delta)$ with one close to its optimal value $\norm{\Pi_{\mathrm{in}}} = \sqrt{\norm{X}^2
+1}$ and the extra inverse-gap factor collapses to $\widetilde{\operatorname{\mathbf{O}}}(1)$. Thus, one directly obtains a near-normalized block-encoding of $X$, and the $\operatorname{\mathbf{O}}(1/\delta)$ cost of amplifying the resulting block-encoding is avoided. The total query complexity to obtain the solution is $\widetilde{\operatorname{\mathbf{O}}}(1/\delta)$ queries---a quadratic improvement over the direct augmented method of~\cite[Remark 6.4]{WL2026}. We discuss this in detail in~\append{app-sylvester}.

%%%%%%%%%%%%%%%%%%%%%%%%%%%%%%%%%%%%%%%%%%%%%%%%%%%%%%%%%%%%%%%%%%%%%%%%%%%%%%
\section{Discussion}
\label{sec:discuss}
We have shown that a quantum computer can perform an oblique eigenprojection $\Pi_{\mathrm{in}}$ given block encoding access to a non-normal input matrix. Our method has a query complexity near linearly in the inverse gap and a normalization factor close to its optimal value $\norm{\Pi_{\mathrm{in}}}$. This improves the standard resolvent-integral approach and matches known results for the orthogonal case. The improvement carries through to concrete applications: preparing eigenstates of matrices with complex eigenvalues, and solving the continuous-time algebraic Riccati and Sylvester equations, where it sharpens the gap dependence of recent quantum algorithms by a quadratic to cubic factor.

Underlying the result is a complete characterization of the block encoding produced by the standard resolvent integral. Rather than tracking only the top left corner that carries the projector, we express every block through the discrete Fourier transform of the resolvent, and it is this full description of the otherwise uninteresting entries that makes the two-sided block preconditioning possible. A similar attention to the entire structure of a block encoding, rather than a single distinguished entry, appears in several works, for reducing the ancilla overhead of block encodings through uncomputing~\cite{vasconcelos2025}, for maintaining matrix encoding in the off-diagonal blocks of a Hamiltonian to compose arithmetic operations~\cite{kangsu2025}, and for constraining the powers of a block encoding to reproduce those of the encoded matrix~\cite{LW18,gutierrez2026}. We expect that a further understanding of this kind will find use in the design of quantum algorithms.

Our preconditioning rests on the closed-form expression for the discrete Fourier transform of the matrix resolvent, and it does not directly apply to every construction. In particular, the main approach of Wang and Liu~\cite{WL2026} represents the matrix sign through a product of two resolvents rather than a single one, for which we do not have an analogous closed form for the discrete Fourier transform, so the two-sided block preconditioning does not carry over. Adapting it to reduce the normalization factor in that setting is an interesting open direction.

It would also be valuable to better understand the case of more general spectral regions. Our construction there first block encodes $f(A)$ for a separating map $f$ and then applies the unit-disk algorithm, so the cost and the spectral-set hypotheses are inherited through $f$. It would be desirable to avoid the explicit block encoding of $f(A)$, and instead work with the resolvent of $A$ itself, integrated over a general curve enclosing the target spectrum. Such a construction could considerably simplify the algorithm, but would require preconditioning directly against the geometry of the region.

A different route to the eigenprojection is through quantum eigenvalue estimation: one estimates the eigenvalues of the input matrix in an ancillary register and then projects by selecting those in the target region. Beyond the QEVT technique from~\cite{QEVP}, a number of algorithms estimate eigenvalues of non-normal or non-Hermitian matrices~\cite{shao2022computing, shao2022solving, zhang2024, alase2024}. This approach, however, requires resolving individual eigenvalues to high precision and it typically imposes additional structural assumptions on spectrum of the input matrix and properties of the initial state. Our construction instead projects by location in the complex plane, requiring only separation of the numerical range from the region boundary, and never estimates an eigenvalue.

A further contrast is worth noting. For a Hermitian operator, one can transform it by QSVT and its cost is solely determined by the function alone: implementing a window function, for instance, costs inversely in its transition width independently of the input matrix. Our oblique eigenprojection algorithm does not admit such a feature. Because the indicator of a spectral region is nonanalytic, it cannot be approximated uniformly and applied blindly; our construction instead controls the resolvent on the boundary of the region, which requires the spectrum of the input to be separated from that boundary by a gap. Our cost is therefore governed by a property of the input matrix---its gap---rather than by a property of the target function. It would be interesting to know whether an eigenprojection can be implemented at a cost depending only on the geometry of the target region, as in the Hermitian case, rather than on the spectral gap of the input.

More broadly, our result suggests a route to applying nonanalytic matrix functions on quantum computers. The oblique eigenprojection corresponds to the indicator of a spectral region, a discontinuous and hence nonanalytic function; our resolvent-based construction sidesteps the analytic obstruction that rules out a direct polynomial or singular value implementation. Whether the same discrete Fourier transform and preconditioning ideas extend to other nonanalytic functions of a non-normal matrix is an intriguing direction for future work.

%%%%%%%%%%%%%%%%%%%%%%%%%%%%%%%%%%%%%%%%%%%%%%%%%%%%%%%%%%%%%%%%%%%%%%%%%%%%%%
\section*{Acknowledgements}
A.M.D.~and Y.S.~thank Fernando Brand\~ao, Oskar Painter, James Hamilton, Nafea Bshara, Peter DeSantis, and Andy Jassy for their involvement and support of the research activities at the AWS Center for Quantum Computing, and Lin Lin for reviewing an earlier draft. 

%%%%%%%%%%%%%%%%%%%%%%%%%%%%%%%%%%%%%%%%%%%%%%%%%%%%%%%%%%%%%%%%%%%%%%%%%%%%%%
\newpage
\appendix
\section{Standard block encoding of resolvent integral}
\label{append:resolvent}
In this appendix, we provide details on the technical results from~\sec{riesz}.

\subsection{Resolvent integral and generalized eigensubspaces}
\label{append:resolvent-integral}

Let $A$ be a square matrix acting on a Hilbert space $\mathcal{H}$. By the primary decomposition theorem~\cite[Theorem 12.2.14]{humpherys2017}\cite[Theorem 7.6]{roman2008advanced}, the entire space splits as a direct sum of generalized eigensubspaces
\begin{equation}
  \label{eq:app-primary}
  \mathcal{H} = \bigoplus_{\lambda \in \operatorname{\mathbf{Spec}}(A)} \mathcal{H}_\lambda,
  \qquad
  \mathcal{H}_\lambda = \operatorname{\mathbf{Ker}}\!\left((A - \lambda I)^{m_\lambda}\right),
\end{equation}
for positive integers $m_\lambda\leq\dim(\mathcal{H})$. Each $\mathcal{H}_\lambda$ is invariant under $A$, and the restriction $N_\lambda := (A - \lambda I)|_{\mathcal{H}_\lambda}$ is nilpotent, with $N_\lambda^{m_\lambda} = 0$.

Let $\mathcal{C}$ be a closed curve that encloses part of the spectrum $\operatorname{\mathbf{Spec}}_{\mathrm{in}}(A)$ and excludes the rest $\operatorname{\mathbf{Spec}}_{\mathrm{out}}(A)$, so that $\operatorname{\mathbf{Spec}}(A) = \operatorname{\mathbf{Spec}}_{\mathrm{in}}(A) \cup \operatorname{\mathbf{Spec}}_{\mathrm{out}}(A)$ with $\operatorname{\mathbf{Spec}}_{\mathrm{in}}(A)$ inside $\mathcal{C}$ and $\operatorname{\mathbf{Spec}}_{\mathrm{out}}(A)$ outside. The Riesz projector associated with $\operatorname{\mathbf{Spec}}_{\mathrm{in}}(A)$ is defined by
\begin{equation}
  \label{eq:app-riesz-contour}
  \Pi_{\mathrm{in}} = \frac{1}{2\pi \mathrm{i}} \int_{\mathcal{C}} (z I - A)^{-1} \, \mathrm{d}z,
\end{equation}
the curve traversed counterclockwise. The following proposition shows that $\Pi_{\mathrm{in}}$ is precisely the oblique projection onto the enclosed generalized eigensubspaces $\bigoplus_{\lambda \in \operatorname{\mathbf{Spec}}_{\mathrm{in}}(A)} \mathcal{H}_\lambda$ along the excluded ones $\bigoplus_{\lambda \in \operatorname{\mathbf{Spec}}_{\mathrm{out}}(A)} \mathcal{H}_\lambda$.

\begin{lemma}
\label{lem:app-riesz-image}
Let $A$ be a square matrix on $\mathcal{H}$ with the generalized eigensubspace decomposition $\mathcal{H} = \bigoplus_{\lambda \in \operatorname{\mathbf{Spec}}(A)} \mathcal{H}_\lambda$ as in \eqref{eq:app-primary}, and let $\Pi_{\mathrm{in}}$ be the Riesz projector \eqref{eq:app-riesz-contour} for a curve $\mathcal{C}$ enclosing $\operatorname{\mathbf{Spec}}_{\mathrm{in}}(A) \subseteq \operatorname{\mathbf{Spec}}(A)$ and excluding $\operatorname{\mathbf{Spec}}_{\mathrm{out}}(A) = \operatorname{\mathbf{Spec}}(A) \setminus \operatorname{\mathbf{Spec}}_{\mathrm{in}}(A)$. Then
\begin{equation}
  \Pi_{\mathrm{in}}^2 = \Pi_{\mathrm{in}},
  \qquad
  \operatorname{\mathbf{Im}}(\Pi_{\mathrm{in}}) = \bigoplus_{\lambda \in \operatorname{\mathbf{Spec}}_{\mathrm{in}}(A)} \mathcal{H}_\lambda,
  \qquad
  \operatorname{\mathbf{Ker}}(\Pi_{\mathrm{in}}) = \bigoplus_{\lambda \in \operatorname{\mathbf{Spec}}_{\mathrm{out}}(A)} \mathcal{H}_\lambda.
\end{equation}
\end{lemma}

\begin{proof}
The resolvent $(zI - A)^{-1}$ preserves the decomposition \eqref{eq:app-primary}, so it suffices to evaluate $\Pi_{\mathrm{in}}$ on each summand $\mathcal{H}_\lambda$. 

Writing $A = \lambda I + N_\lambda$ on $\mathcal{H}_\lambda$, for $z \notin \operatorname{\mathbf{Spec}}(A)$ the restricted resolvent is the finite Neumann series
\begin{equation}
  \label{eq:resolvent-neumann}
  (z I - A)^{-1}\big|_{\mathcal{H}_\lambda}
  = \left((z - \lambda) I - N_\lambda\right)^{-1}
  = \sum_{k=0}^{m_\lambda - 1} \frac{N_\lambda^{\,k}}{(z - \lambda)^{k+1}},
\end{equation}
terminating because $N_\lambda^{m_\lambda} = 0$. Integrating term by term over $\mathcal{C}$,
\begin{equation}
  \label{eq:app-winding}
  \frac{1}{2\pi \mathrm{i}} \int_{\mathcal{C}} \frac{\mathrm{d}z}{(z - \lambda)^{k+1}}
  = \begin{cases} 1 & k = 0 \text{ and } \lambda \in \operatorname{\mathbf{Spec}}_{\mathrm{in}}(A), \\
  0 & \text{otherwise},\end{cases}
\end{equation}
since the integrand has vanishing residue for $k \ge 1$, and for $k = 0$ the integral is the winding number of $\mathcal{C}$ about $\lambda$. 

Thus $\Pi_{\mathrm{in}}$ restricts to the identity on $\mathcal{H}_\lambda$ for $\lambda \in \operatorname{\mathbf{Spec}}_{\mathrm{in}}(A)$ and to $0$ for $\lambda \in \operatorname{\mathbf{Spec}}_{\mathrm{out}}(A)$. The first of these gives $\mathcal{H}_\lambda \subseteq \operatorname{\mathbf{Im}}(\Pi_{\mathrm{in}})$ and $\Pi_{\mathrm{in}}^2 = \Pi_{\mathrm{in}}$ on $\mathcal{H}_\lambda$; the second gives $\mathcal{H}_\lambda \subseteq \operatorname{\mathbf{Ker}}(\Pi_{\mathrm{in}})$ and again $\Pi_{\mathrm{in}}^2 = \Pi_{\mathrm{in}}$. Summing over the decomposition \eqref{eq:app-primary} yields $\Pi_{\mathrm{in}}^2 = \Pi_{\mathrm{in}}$, $\operatorname{\mathbf{Im}}(\Pi_{\mathrm{in}}) = \bigoplus_{\lambda \in \operatorname{\mathbf{Spec}}_{\mathrm{in}}(A)} \mathcal{H}_\lambda$, and $\operatorname{\mathbf{Ker}}(\Pi_{\mathrm{in}}) = \bigoplus_{\lambda \in \operatorname{\mathbf{Spec}}_{\mathrm{out}}(A)} \mathcal{H}_\lambda$.
\end{proof}

\subsection{Block encoding of resolvent integral}
\label{append:resolvent-blockencoding}
\append{resolvent-integral} treated a general curve $\mathcal{C}$. We now specialize to the unit circle. Parametrizing $z = \mathrm{e}^{\mathrm{i}\theta}$ in \eqref{eq:app-riesz-contour}, with $\mathrm{d}z = \mathrm{i} \mathrm{e}^{\mathrm{i}\theta}\, \mathrm{d}\theta$ and $\mathrm{e}^{\mathrm{i}\theta}(\mathrm{e}^{\mathrm{i}\theta} I - A)^{-1} = (I - \mathrm{e}^{-\mathrm{i}\theta}A)^{-1}$, the Riesz projector \eqref{eq:app-riesz-contour} becomes
\begin{equation}
  \label{eq:app-riesz-unitcircle}
  \Pi_{\mathrm{in}} 
  = \frac{1}{2\pi \mathrm{i}} \int_{\abs{z}=1} (z I - A)^{-1} \, \mathrm{d}z
  = \frac{1}{2\pi} \int_0^{2\pi} (I - \mathrm{e}^{-\mathrm{i}\theta}A)^{-1} \, \mathrm{d}\theta.
\end{equation}
This can be seen as the zeroth Fourier coefficient of the periodic function $\theta \mapsto (I - \mathrm{e}^{-\mathrm{i}\theta}A)^{-1}$.

We discretize the integral \eqref{eq:app-riesz-unitcircle} by the $n$-point trapezoidal rule at the roots of unity $\omega_n^j$, where $\omega_n = \mathrm{e}^{2\pi \mathrm{i} / n}$ and $j = 0, \dots, n-1$,
\begin{equation}
  \label{eq:app-trap}
  \Pi_{\mathrm{in}} \approx S_0 := \frac{1}{n} \sum_{j=0}^{n-1} (I - \omega_n^{-j} A)^{-1}.
\end{equation}
Because the integrand is analytic on an annulus containing the unit circle, this approximation converges geometrically in $n$, with a rate we quantify in \append{precond}. The node count $n$ enters only through the size $\lceil \log_2 n \rceil$ of the ancilla register indexing the roots of unity, and does not affect the query complexity of the construction below, which depends only on the normalization factor of the block encoding and the target precision.

The discretization \eqref{eq:app-trap} expresses $\Pi_{\mathrm{in}}$ through resolvents of $A$, which can be realized on a quantum computer via block encodings. We recall standard block-encoding operations used throughout, referring to~\cite{gilyen2019, QEVP} for details. In all complexity statements, the parameters are known bounds on the underlying quantities, such as upper bounds on normalization factors, norms, and spectral-set constants, and lower bounds on gaps, overlaps, and smallest nonzero singular values. The algorithms use these bounds to choose their internal parameters, and the stated costs hold with the bounds in place of the exact values.

\begin{lemma}[Linear combination]
\label{lem:app-be-lincomb}
Given block encodings of matrices $A_1/\alpha_1, \dots, A_m/\alpha_m$ with normalization factors $\alpha_1, \dots, \alpha_m>0$, and coefficients $c_1, \dots, c_m$, one can construct a block encoding of
\begin{equation}
    \frac{1}{\sum_{k=1}^m \abs{c_k}\alpha_k}\sum_{k=1}^m c_kA_k
\end{equation}
with normalization factor $\sum_{k=1}^m \abs{c_k}\alpha_k$, using a single query to a controlled unitary that selects among the input block encodings, whose cost is the maximum of their individual query costs.
\end{lemma}

\begin{lemma}[Product]
\label{lem:app-be-product}
Given block encodings of matrices $A_1/\alpha_1$ and $A_2/\alpha_2$ with normalization factors $\alpha_1, \alpha_2 > 0$, one can construct a block encoding of
\begin{equation}
    \frac{1}{\alpha_1 \alpha_2} A_1 A_2
\end{equation}
with normalization factor $\alpha_1 \alpha_2$, using a single query to each of the two input block encodings.
\end{lemma}

\begin{lemma}[Inversion]
\label{lem:app-be-inversion}
Given a block encoding of $A/\alpha_A$ with normalization factor $\alpha_A > 0$, and an upper bound $\alpha_{A^{-1}} \geq \norm{A^{-1}}$ on the inverse, for any $\varepsilon > 0$ one can construct a block encoding of
\begin{equation}
    \frac{1}{2 \alpha_{A^{-1}}} A^{-1}
\end{equation}
with normalization factor $2 \alpha_{A^{-1}}$ and accuracy $\varepsilon$, using  
\begin{equation}
    \operatorname{\mathbf{O}}\!\left(\alpha_A \alpha_{A^{-1}} \log\!\frac{1}{\varepsilon}\right)
\end{equation}
queries to the block encoding of $A/\alpha_A$.
\end{lemma}

\begin{lemma}[Scaling]
\label{lem:app-be-scaling}
Given a block encoding of $A/\alpha_A$ with normalization factor $\alpha_A > 0$, and a value $\alpha_A' \geq \norm{A}$ upper bounding the norm:
\begin{enumerate}
  \item if $\alpha_A' \geq \alpha_A$, one can construct a block encoding of
    \begin{equation}
        \frac{1}{\alpha_A'} A
    \end{equation}
    with normalization factor $\alpha_A'$, using a single query to the block encoding of $A/\alpha_A$;
  \item if $\alpha_A' < \alpha_A$, then for any $\varepsilon > 0$ one can construct a block encoding of
    \begin{equation}
        \frac{1}{2\alpha_A'} A
    \end{equation}
    with normalization factor $2\alpha_A'$ and accuracy $\varepsilon$, using
    \begin{equation}
        \operatorname{\mathbf{O}}\!\left(\frac{\alpha_A}{\alpha_A'} \log\!\frac{1}{\varepsilon}\right)
    \end{equation}
    queries to the controlled block encoding of $A/\alpha_A$ and its inverse.
\end{enumerate}
\end{lemma}

When applied to implementing the resolvent integral, the standard block encoding approach has a normalization factor governed by the maximum size of the resolvent on the unit circle
\begin{equation}
    \alpha_{\mathrm{res}}=\max_{\theta \in [0, 2\pi)} \norm{\left(I - \mathrm{e}^{-\mathrm{i}\theta} A\right)^{-1}},
\end{equation}
which can be bounded under two standard assumptions on $A$. We denote the compressions onto the inner and outer invariant subspaces \eqref{eq:app-primary} by $A_{\mathrm{in}} = A\Pi_{\mathrm{in}}$ and $A_{\mathrm{out}} = A\Pi_{\mathrm{out}}$.

The first assumption controls the resolvent through the numerical range. Recall that for a matrix $B$, $\mathcal{W}(B) = \{\, \bra{\psi} B \ket{\psi} : \norm{\psi} = 1 \,\}$ is its numerical range. We assume the per-block numerical ranges are separated from the unit circle,
\begin{equation}
  \label{eq:app-nr-assumption}
  \mathcal{W}(A_{\mathrm{in}}) \subseteq \{\, z \colon \abs{z} \leq 1 - \delta_{\mathrm{num}} \,\},
  \qquad
  \mathcal{W}(A_{\mathrm{out}}^{-1}) \subseteq \{\, z \colon \abs{z} \leq 1 - \delta_{\mathrm{num}} \,\},
\end{equation}
for some $0 < \delta_{\mathrm{num}} < 1$. The Crouzeix--Palencia bound~\cite{CrouzeixPalencia17} then applies to each block, giving $\norm{g(A_{\mathrm{in}})} \leq (1 + \sqrt{2}) \max_{z \in \mathcal{W}(A_{\mathrm{in}})} \abs{g(z)}$ for any $g$ analytic on a neighborhood of $\mathcal{W}(A_{\mathrm{in}})$, and likewise for
$A_{\mathrm{out}}^{-1}$.

The second assumption is diagonalizability. We assume the compressions $A_{\mathrm{in}} = V_{\mathrm{in}} \Lambda_{\mathrm{in}} V_{\mathrm{in}}^{-1}$ and $A_{\mathrm{out}} = V_{\mathrm{out}} \Lambda_{\mathrm{out}} V_{\mathrm{out}}^{-1}$ are respectively diagonalizable, so that $\norm{g(A_{\mathrm{in}})} \leq \kappa_V \max_{\lambda \in \operatorname{\mathbf{Spec}}(A_{\mathrm{in}})} \abs{g(\lambda)}$ for any $g$ analytic on a neighborhood of $\operatorname{\mathbf{Spec}}(A_{\mathrm{in}})$, and likewise for $A_{\mathrm{out}}$, where 
\begin{equation}
  \label{eq:app-kappa-V}
  \kappa_V = \max\!\left(\norm{V_{\mathrm{in}}}\norm{V_{\mathrm{in}}^{-1}},\;
  \norm{V_{\mathrm{out}}}\norm{V_{\mathrm{out}}^{-1}}\right),\qquad
  \delta_{\mathrm{eig}} = \min_{\lambda \in \operatorname{\mathbf{Spec}}(A)} \big|\, \abs{\lambda} - 1 \,\big|
\end{equation}
is the per-block eigenbasis condition number replacing the Crouzeix--Palencia constant. This analysis extends to generic matrices through the Jordan condition number~\cite[Page 444]{TrefethenEmbree05}.

Both assumptions on the input matrices are instances of a single condition, which we now formalize.

\begin{definition}[Spectral set]
\label{def:app-spectral-set}
Let $B$ be a square matrix and let $\mathcal{K} \subseteq \mathbb{C}$ be a compact set containing $\operatorname{\mathbf{Spec}}(B)$. We call $\mathcal{K}$ a spectral set for $B$ with constant $c \geq 1$ if
\begin{equation}
  \label{eq:app-spectral-set}
  \norm{g(B)} \leq c \sup_{z \in \mathcal{K}} \abs{g(z)}
\end{equation}
for every function $g$ analytic on a neighborhood of $\mathcal{K}$.
\end{definition}

\begin{lemma}[Instances of spectral sets]
\label{lem:app-spectral-instances}
Let $B$ be a square matrix.
\begin{enumerate}
  \item The numerical range $\mathcal{W}(B)$ is a spectral set for $B$ with constant
    $c = 1 + \sqrt{2}$.
  \item If $B = V \Lambda V^{-1}$ with $\Lambda$ diagonal, then $\operatorname{\mathbf{Spec}}(B)$ is a spectral set for
    $B$ with constant $c = \norm{V}\norm{V^{-1}}$.
\end{enumerate}
\end{lemma}

\begin{proof}
Part (1) is the Crouzeix--Palencia bound~\cite{CrouzeixPalencia17}. For part (2), any $g$ analytic on a neighborhood of $\operatorname{\mathbf{Spec}}(B)$ satisfies $g(B) = V g(\Lambda) V^{-1}$. Since $g(\Lambda)$ is diagonal with entries $g(\lambda)$, $\lambda \in \operatorname{\mathbf{Spec}}(B)$,
\begin{equation}
  \norm{g(B)} \leq \norm{V}\norm{V^{-1}} \max_{\lambda \in \operatorname{\mathbf{Spec}}(B)} \abs{g(\lambda)}.
\end{equation}
\end{proof}

We say that the compressions $A_{\mathrm{in}}$ and $A_{\mathrm{out}}^{-1}$ admit spectral sets $\mathcal{K}_{\mathrm{in}}$ and $\mathcal{K}_{\mathrm{out}}$ \eqref{eq:app-spectral-set} with constant $c$ and gap $\delta$ if both sets have the common constant $c$ and satisfy
\begin{equation}
  \label{eq:app-spectral-hyp}
  \mathcal{K}_{\mathrm{in}} \subseteq \{\, z \colon \abs{z} \leq 1 - \delta \,\},
  \qquad
  \mathcal{K}_{\mathrm{out}} \subseteq \{\, z \colon \abs{z} \leq 1 - \delta \,\}
\end{equation}
for some $0 < \delta < 1$. Here, for a compression such as $A_{\mathrm{in}} = A\Pi_{\mathrm{in}}$, spectra, numerical ranges, and spectral sets are taken for its restriction to the invariant subspace $\mathcal{H}_{\mathrm{in}}$, and similarly for $A_{\mathrm{out}}$ and $A_{\mathrm{out}}^{-1}$ on $\mathcal{H}_{\mathrm{out}}$. By \lem{app-spectral-instances}, this condition holds in two cases. Under the numerical-range assumption \eqref{eq:app-nr-assumption}, take $\mathcal{K}_{\mathrm{in}} = \mathcal{W}(A_{\mathrm{in}})$ and $\mathcal{K}_{\mathrm{out}} = \mathcal{W}(A_{\mathrm{out}}^{-1})$, with $c = 1+\sqrt{2}$ and $\delta = \delta_{\mathrm{num}}$. Under the diagonalizability assumption, take $\mathcal{K}_{\mathrm{in}} = \operatorname{\mathbf{Spec}}(A_{\mathrm{in}})$ and $\mathcal{K}_{\mathrm{out}} = \operatorname{\mathbf{Spec}}(A_{\mathrm{out}}^{-1})$, with $c = \kappa_V$ \eqref{eq:app-kappa-V}. Every eigenvalue of $A_{\mathrm{out}}$ has modulus at least $1 + \delta_{\mathrm{eig}}$. So the eigenvalues of $A_{\mathrm{out}}^{-1}$ have modulus at most $1/(1+\delta_{\mathrm{eig}}) \leq 1 - \delta_{\mathrm{eig}}/2$, and the condition holds with $\delta = \delta_{\mathrm{eig}}/2$.

\begin{proposition}[Resolvent norm bound]
\label{prop:app-resolvent-bound}
Let $A$ be a square matrix with Riesz projector $\Pi_{\mathrm{in}}$ \eqref{eq:app-riesz-unitcircle} and compressions $A_{\mathrm{in}} = A\Pi_{\mathrm{in}}$, $A_{\mathrm{out}} = A\Pi_{\mathrm{out}}$. If $A_{\mathrm{in}}$ and $A_{\mathrm{out}}^{-1}$ admit spectral sets \eqref{eq:app-spectral-set} with constant $c$ and gap $\delta$ satisfying \eqref{eq:app-spectral-hyp}, then the maximum resolvent norm satisfies
\begin{equation}
  \label{eq:app-alpha-res}
  \alpha_{\mathrm{res}} = \max_{\theta \in [0,2\pi)} \norm{(I - \mathrm{e}^{-\mathrm{i}\theta}A)^{-1}}
  \leq \frac{2c\,\norm{\Pi_{\mathrm{in}}}}{\delta}.
\end{equation}
Moreover:
\begin{enumerate}
  \item under the numerical-range assumption \eqref{eq:app-nr-assumption} with gap
    $\delta_{\mathrm{num}}$,
    \begin{equation}
      \label{eq:app-alpha-res-num}
      \alpha_{\mathrm{res}} \leq \frac{2(1 + \sqrt{2})\,\norm{\Pi_{\mathrm{in}}}}{\delta_{\mathrm{num}}}.
    \end{equation}
    \item if $A_{\mathrm{in}}$ and $A_{\mathrm{out}}$ are diagonalizable with per-block eigenbasis
    condition number $\kappa_V$ and spectrum gap $\delta_{\mathrm{eig}}$ \eqref{eq:app-kappa-V},
    \begin{equation}
      \label{eq:app-alpha-res-eig}
      \alpha_{\mathrm{res}} \leq \frac{2\,\kappa_V\,\norm{\Pi_{\mathrm{in}}}}{\delta_{\mathrm{eig}}}.
    \end{equation}
\end{enumerate}
\end{proposition}

\begin{proof}
The resolvent splits along the invariant subspaces as
\begin{equation}
  \label{eq:app-resolvent-split}
  (I - \mathrm{e}^{-\mathrm{i}\theta}A)^{-1}
  = (I - \mathrm{e}^{-\mathrm{i}\theta}A_{\mathrm{in}})^{-1}\Pi_{\mathrm{in}}
  + (I - \mathrm{e}^{-\mathrm{i}\theta}A_{\mathrm{out}})^{-1}\Pi_{\mathrm{out}},
\end{equation}
so that
\begin{equation}
  \label{eq:app-resolvent-split-norm}
  \norm{(I - \mathrm{e}^{-\mathrm{i}\theta}A)^{-1}}
  \leq \norm{(I - \mathrm{e}^{-\mathrm{i}\theta}A_{\mathrm{in}})^{-1}}\,\norm{\Pi_{\mathrm{in}}}
  + \norm{(I - \mathrm{e}^{-\mathrm{i}\theta}A_{\mathrm{out}})^{-1}}\,\norm{\Pi_{\mathrm{out}}}.
\end{equation}
For an oblique projector and its complement, we have the tight bound $\norm{\Pi_{\mathrm{out}}} = \norm{I - \Pi_{\mathrm{in}}} = \norm{\Pi_{\mathrm{in}}}$ by \lem{app-projector-norm} below, so it remains to bound the two restricted resolvents. For the purpose of deriving asymptotic bounds, the looser estimate $\norm{\Pi_{\mathrm{out}}} \leq 1 + \norm{\Pi_{\mathrm{in}}}$ from triangle inequality would also suffice.

For the inner block, take $g(z) = (1 - \mathrm{e}^{-\mathrm{i}\theta}z)^{-1}$. Its only pole, $z = \mathrm{e}^{\mathrm{i}\theta}$, lies outside $\mathcal{K}_{\mathrm{in}}$, and $\abs{1 - \mathrm{e}^{-\mathrm{i}\theta}z} \geq \delta$ on $\mathcal{K}_{\mathrm{in}}$. By \eqref{eq:app-spectral-set}, $\norm{(I - \mathrm{e}^{-\mathrm{i}\theta}A_{\mathrm{in}})^{-1}} \leq c/\delta$.

For the outer block, write $(I - \mathrm{e}^{-\mathrm{i}\theta}A_{\mathrm{out}})^{-1} = h(A_{\mathrm{out}}^{-1})$ with $h(w) = w(w - \mathrm{e}^{-\mathrm{i}\theta})^{-1}$. Its only pole, $w = \mathrm{e}^{-\mathrm{i}\theta}$, lies outside $\mathcal{K}_{\mathrm{out}}$, and $\abs{h(w)} \leq (1-\delta)/\delta \leq 1/\delta$ on $\mathcal{K}_{\mathrm{out}}$. By \eqref{eq:app-spectral-set}, $\norm{(I - \mathrm{e}^{-\mathrm{i}\theta}A_{\mathrm{out}})^{-1}} \leq c/\delta$.

Combining the two blocks gives \eqref{eq:app-alpha-res}. By \lem{app-spectral-instances}, \eqref{eq:app-spectral-hyp} holds with $c = 1+\sqrt{2}$ and $\delta = \delta_{\mathrm{num}}$ under the numerical-range assumption, which gives \eqref{eq:app-alpha-res-num}. Under the diagonalizability assumption, it holds with $c = \kappa_V$ and $\delta = \delta_{\mathrm{eig}}/2$, which gives \eqref{eq:app-alpha-res-eig} with prefactor $4\kappa_V$.

In the diagonalizable case, the prefactor can be improved to $2\kappa_V$ by bounding the outer block directly rather than through $A_{\mathrm{out}}^{-1}$. For every $\lambda \in \operatorname{\mathbf{Spec}}(A)$,
\begin{equation}
  \abs{1 - \mathrm{e}^{-\mathrm{i}\theta}\lambda} = \abs{\mathrm{e}^{\mathrm{i}\theta} - \lambda}
  \geq \big|\abs{\lambda} - 1\big| \geq \delta_{\mathrm{eig}},
\end{equation}
whether $\lambda$ lies inside or outside the unit circle. Applying the eigenbasis bound of \lem{app-spectral-instances} to $g(z) = (1 - \mathrm{e}^{-\mathrm{i}\theta}z)^{-1}$ on each block gives $\norm{(I - \mathrm{e}^{-\mathrm{i}\theta}A_{\mathrm{in}})^{-1}}, \norm{(I - \mathrm{e}^{-\mathrm{i}\theta}A_{\mathrm{out}})^{-1}} \leq \kappa_V/\delta_{\mathrm{eig}}$. Hence $\alpha_{\mathrm{res}} \leq 2\kappa_V\,\norm{\Pi_{\mathrm{in}}}/\delta_{\mathrm{eig}}$.
\end{proof}

\begin{lemma}[Norm of an oblique projection~\cite{szyld2006}]
\label{lem:app-projector-norm}
Let $\Pi^2 = \Pi$ be an oblique projection, with $\Pi \neq 0$ and $\Pi \neq I$. Then
$\norm{\Pi} = \norm{I - \Pi}$.
\end{lemma}

\begin{proof}
The entire space admits the orthogonal decomposition $\mathcal{H} = \operatorname{\mathbf{Im}}(\Pi) \oplus \operatorname{\mathbf{Im}}(\Pi)^\perp$. Every vector in $\operatorname{\mathbf{Im}}(\Pi)$ is fixed by $\Pi$, and $\Pi$ maps $\operatorname{\mathbf{Im}}(\Pi)^\perp$ into $\operatorname{\mathbf{Im}}(\Pi)$, so with respect to this splitting
\begin{equation}
  \Pi = \begin{bmatrix} I & P \\ 0 & 0 \end{bmatrix},
  \qquad
  I - \Pi = \begin{bmatrix} 0 & -P \\ 0 & I \end{bmatrix},
\end{equation}
for some operator $P : \operatorname{\mathbf{Im}}(\Pi)^\perp \to \operatorname{\mathbf{Im}}(\Pi)$. Then
\begin{equation}
  \Pi \Pi^\dagger = \begin{bmatrix} I + P P^\dagger & 0 \\ 0 & 0 \end{bmatrix},
  \qquad
  (I - \Pi)^\dagger (I - \Pi) = \begin{bmatrix} 0 & 0 \\ 0 & I + P^\dagger P \end{bmatrix},
\end{equation}
so $\norm{\Pi}^2 = \norm{\Pi \Pi^\dagger} = 1 + \norm{P P^\dagger} = 1 + \norm{P}^2$ and $\norm{I - \Pi}^2 = \norm{(I - \Pi)^\dagger (I - \Pi)} = 1 + \norm{P^\dagger P} = 1 + \norm{P}^2$. Hence $\norm{\Pi} = \norm{I - \Pi}$.
\end{proof}

We now assemble the standard block encoding. Introduce an ancilla register indexing the roots of unity and form the block-diagonal matrix
\begin{equation}
  \label{eq:app-M-blockdiag}
  M = \sum_{j=0}^{n-1} \ketbra{j}{j} \otimes (I - \omega_n^{-j} A),
\end{equation}
whose blocks are the shifted matrices appearing in the discretization \eqref{eq:app-trap}. Each block is the combination $I - \omega_n^{-j} A$ of $I$ and $A$, so \lem{app-be-lincomb} block encodes $M$ with normalization factor $\alpha_A + 1 = \operatorname{\mathbf{O}}(\alpha_A)$ from the block encoding of $A/\alpha_A$. Applying the inversion of \lem{app-be-inversion} block encodes $M^{-1} = \sum_{j=0}^{n-1} \ketbra{j}{j} \otimes (I - \omega_n^{-j} A)^{-1}$, and preparing the ancilla in the uniform superposition $\operatorname{\mathbf{Had}}\ket{0} = \tfrac{1}{\sqrt{n}} \sum_{j=0}^{n-1} \ket{j}$ extracts the average
\begin{equation}
  \label{eq:app-hadamard-extract}
  \left(\bra{0}\operatorname{\mathbf{Had}} \otimes I\right) M^{-1}
  \left(\operatorname{\mathbf{Had}}\ket{0} \otimes I\right)
  = \frac{1}{n} \sum_{j=0}^{n-1} (I - \omega_n^{-j} A)^{-1} = S_0.
\end{equation}

\begin{proposition}[Standard block encoding of the oblique eigenprojection]
\label{prop:app-standard-be}
Let $A/\alpha_A$ be a block encoding of a square matrix $A$ with Riesz projector $\Pi_{\mathrm{in}} = \frac{1}{2\pi \mathrm{i}} \int_{\abs{z}=1} (z I - A)^{-1} \, \mathrm{d}z$, and let $\alpha_{\mathrm{res}} = \max_{\theta \in [0,2\pi)} \norm{(I - \mathrm{e}^{-\mathrm{i}\theta}A)^{-1}}$ be the maximum resolvent norm. For any $\varepsilon > 0$, $\Pi_{\mathrm{in}}$ can be block encoded with normalization factor $\operatorname{\mathbf{O}}(\alpha_{\mathrm{res}})$ and accuracy $\varepsilon$, using
\begin{equation}
  \label{eq:app-standard-queries}
  \operatorname{\mathbf{O}}\!\left(\alpha_A\, \alpha_{\mathrm{res}}
  \log\frac{1}{\varepsilon}\right)
\end{equation}
queries to the block encoding of $A/\alpha_A$. 

If the compressions $A_{\mathrm{in}}$ and $A_{\mathrm{out}}^{-1}$ admit spectral sets \eqref{eq:app-spectral-set} with constant $c$ and gap $\delta$ satisfying \eqref{eq:app-spectral-hyp}, then $\alpha_{\mathrm{res}} = \operatorname{\mathbf{O}}(c\,\norm{\Pi_{\mathrm{in}}}/\delta)$ by \prop{app-resolvent-bound}. In particular:
\begin{enumerate}
  \item under the numerical-range assumption \eqref{eq:app-nr-assumption},
    $\alpha_{\mathrm{res}} = \operatorname{\mathbf{O}}(\norm{\Pi_{\mathrm{in}}}/\delta_{\mathrm{num}})$;
  \item under the diagonalizability assumption \eqref{eq:app-kappa-V},
    $\alpha_{\mathrm{res}} = \operatorname{\mathbf{O}}(\kappa_V \norm{\Pi_{\mathrm{in}}}/\delta_{\mathrm{eig}})$.
\end{enumerate}
\end{proposition}

\begin{proof}
The block encoding of $M$ \eqref{eq:app-M-blockdiag} follows from \lem{app-be-lincomb}, each block being the combination $I - \omega_n^{-j}A$ of $I$ and $A$, with normalization factor $\alpha_A + 1 = \operatorname{\mathbf{O}}(\alpha_A)$. The nonzero singular values of $M$ are bounded below by $1/\alpha_{\mathrm{res}}$, since $\norm{M^{-1}} = \max_j \norm{(I - \omega_n^{-j}A)^{-1}} \leq \alpha_{\mathrm{res}}$, so \lem{app-be-inversion} block encodes $M^{-1}$ with the query count \eqref{eq:app-standard-queries} and normalization factor $\operatorname{\mathbf{O}}(\alpha_{\mathrm{res}})$. The Hadamard sandwich \eqref{eq:app-hadamard-extract} is a linear combination of the diagonal blocks, which preserves the normalization factor and adds only the ancilla preparation. The resulting block encoding is that of $S_0$, which is $\varepsilon$-close to $\Pi_{\mathrm{in}}$ by \eqref{eq:app-trap} once $n$ is chosen as in \append{precond}. The two specializations are the bounds of \prop{app-resolvent-bound}.
\end{proof}

%%%%%%%%%%%%%%%%%%%%%%%%%%%%%%%%%%%%%%%%%%%%%%%%%%%%%%%%%%%%%%%%%%%%%%%%%%%%%%
\section{Discrete Fourier transform of resolvent and preconditioning}
\label{append:precond}

In this appendix, we provide details on the technical results from~\sec{precond}.

\subsection{Block encoding matrix resolvent in the Fourier basis}
\label{append:precond-dft}

Recall the block-diagonal matrix \eqref{eq:app-M-blockdiag} whose inverse collects the shifted resolvents,
\begin{equation}
  M^{-1} = \sum_{j=0}^{n-1} \ketbra{j}{j} \otimes (I - \omega_n^{-j} A)^{-1}.
\end{equation}
Conjugating by the quantum Fourier transform $\operatorname{\mathbf{QFT}}$, acting on the ancilla register as $\operatorname{\mathbf{QFT}} \ket{j} = \tfrac{1}{\sqrt{n}} \sum_{k=0}^{n-1}
\omega_n^{jk} \ket{k}$, its $(k,l)$ block is
\begin{equation}
  \label{eq:app-kl-block}
  \begin{aligned}
    (\bra{k} \otimes I)\, (\operatorname{\mathbf{QFT}} \otimes I)\, M^{-1}\,
      (\operatorname{\mathbf{QFT}}^\dagger \otimes I)\, (\ket{l} \otimes I)
    &= \sum_{j=0}^{n-1} \bra{k}\operatorname{\mathbf{QFT}}\ketbra{j}{j}\operatorname{\mathbf{QFT}}^\dagger\ket{l}
       \, (I - \omega_n^{-j} A)^{-1} \\
    &= \sum_{j=0}^{n-1} \frac{\omega_n^{jk}}{\sqrt{n}} \cdot \frac{\omega_n^{-jl}}{\sqrt{n}}\,
       (I - \omega_n^{-j} A)^{-1} \\
    &= \frac{1}{n} \sum_{j=0}^{n-1} \omega_n^{j(k-l)} (I - \omega_n^{-j} A)^{-1}.
  \end{aligned}
\end{equation}

This depends only on $k - l \bmod n$, so defining the discrete Fourier transform of the resolvent
\begin{equation}
  \label{eq:app-mode-def}
  S_m = \frac{1}{n} \sum_{j=0}^{n-1} \omega_n^{jm} (I - \omega_n^{-j} A)^{-1},
  \qquad m = 0, \dots, n-1,
\end{equation}
the conjugated operator is block circulant,
\begin{equation}
  \label{eq:app-block-circulant}
  \left( \operatorname{\mathbf{QFT}} \otimes I \right) M^{-1}
    \left( \operatorname{\mathbf{QFT}}^\dagger \otimes I \right)
  =
  \begin{bmatrix}
    S_0     & S_{n-1}  & S_{n-2}  & \cdots & S_1      \\
    S_1     & S_0      & S_{n-1}  & \cdots & S_2      \\
    S_2     & S_1      & S_0      & \cdots & S_3      \\
    \vdots  & \vdots   & \vdots   & \ddots & \vdots   \\
    S_{n-1} & S_{n-2}  & S_{n-3}  & \cdots & S_0
  \end{bmatrix}.
\end{equation}
The zeroth mode $S_0$ is the trapezoidal discretization \eqref{eq:app-trap} of the oblique eigenprojection $\Pi_{\mathrm{in}}$. The standard approach of \append{resolvent} extracts $S_0$ alone through the Hadamard conjugation \eqref{eq:app-hadamard-extract}, leaving the remaining blocks undetermined. The Fourier transform \eqref{eq:app-block-circulant} instead characterizes the entire operator with Fourier modes $S_1, \dots, S_{n-1}$, which is what makes the preconditioning below possible.

\subsection{Discrete Fourier transform of matrix resolvent}
\label{append:precond-modes}

The $S_m$ of \eqref{eq:app-mode-def} are related to the continuous Fourier coefficients of the
resolvent,
\begin{equation}
  \label{eq:app-cont-coeff}
  C_m = \frac{1}{2\pi} \int_0^{2\pi} \mathrm{e}^{\mathrm{i}m\theta} (I - \mathrm{e}^{-\mathrm{i}\theta}A)^{-1} \, \mathrm{d}\theta,
  \qquad m \in \mathbb{Z},
\end{equation}
by aliasing. This aliasing relation is a form of the Poisson summation formula; we give an elementary proof using only the orthogonality of roots of unity.

\begin{lemma}[Aliasing identity]
\label{lem:app-aliasing}
Fixing $n$, the discrete Fourier transform \eqref{eq:app-mode-def} of the resolvent and its continuous Fourier coefficients \eqref{eq:app-cont-coeff} satisfy
\begin{equation}
  \label{eq:app-aliasing}
  S_m = \sum_{q \in \mathbb{Z}} C_{m + qn}.
\end{equation}
\end{lemma}

\begin{proof}
Since the resolvent $\theta \mapsto (I - \mathrm{e}^{-\mathrm{i}\theta}A)^{-1}$ is analytic on an annulus containing the unit circle, its Fourier series converges absolutely,
\begin{equation}
  \label{eq:app-resolvent-fourier}
  (I - \mathrm{e}^{-\mathrm{i}\theta}A)^{-1} = \sum_{p \in \mathbb{Z}} C_p\, \mathrm{e}^{-\mathrm{i}p\theta},
\end{equation}
with the coefficients $C_p$ of \eqref{eq:app-cont-coeff}. Evaluating \eqref{eq:app-resolvent-fourier} at the nodes $\theta_j = 2\pi j/n$, where $\mathrm{e}^{-\mathrm{i}\theta_j} = \omega_n^{-j}$, and substituting into the definition \eqref{eq:app-mode-def},
\begin{equation}
  S_m
  = \frac{1}{n} \sum_{j=0}^{n-1} \omega_n^{jm} \sum_{p \in \mathbb{Z}} C_p\, \omega_n^{-jp}
  = \sum_{p \in \mathbb{Z}} C_p \left( \frac{1}{n} \sum_{j=0}^{n-1} \omega_n^{j(m-p)} \right),
\end{equation}
where the interchange of the two sums is justified by the absolute convergence of \eqref{eq:app-resolvent-fourier}. The inner sum is the orthogonality relation of the roots of unity,
\begin{equation}
  \label{eq:app-root-orthogonality}
  \frac{1}{n} \sum_{j=0}^{n-1} \omega_n^{j(m-p)}
  = \begin{cases} 1 & p \equiv m \pmod{n}, \\ 0 & \text{otherwise},\end{cases}
\end{equation}
which retains exactly the terms $p = m + qn$, $q \in \mathbb{Z}$, giving $S_m = \sum_{q \in \mathbb{Z}} C_{m+qn}$.
\end{proof}

By \lem{app-aliasing}, the discrete Fourier transform is determined by the continuous coefficients $C_k$ through aliasing, which we now compute explicitly on each invariant subspace.

\begin{lemma}[Continuous Fourier transform of matrix resolvent]
\label{lem:app-continuous}
Let $A$ be a square matrix whose spectrum is disjoint from the unit circle, with Riesz projector $\Pi_{\mathrm{in}}$ \eqref{eq:app-riesz-unitcircle}, complementary projector $\Pi_{\mathrm{out}} = I - \Pi_{\mathrm{in}}$, and compressions $A_{\mathrm{in}} = A\Pi_{\mathrm{in}}$, $A_{\mathrm{out}} = A\Pi_{\mathrm{out}}$. The continuous Fourier coefficients $C_k$ \eqref{eq:app-cont-coeff} of the resolvent are
\begin{equation}
  \label{eq:app-Ck}
  C_k =
  \begin{cases}
    A_{\mathrm{in}}^{k} & k \geq 1, \\
    \Pi_{\mathrm{in}} & k = 0, \\
    -A_{\mathrm{out}}^{k} & k \leq -1,
  \end{cases}
\end{equation}
where $A_{\mathrm{out}}^{k}$ for $k \leq -1$ denotes a power of the inverse of $A_{\mathrm{out}}$ on $\mathcal{H}_{\mathrm{out}}$, extended by zero on $\mathcal{H}_{\mathrm{in}}$.
\end{lemma}

\begin{proof}
On $\mathcal{H}_{\mathrm{in}}$ all eigenvalues of $A_{\mathrm{in}}$ have modulus below $1$, so the integrand of \eqref{eq:app-cont-coeff} expands as the norm-convergent series
\begin{equation}
    \mathrm{e}^{\mathrm{i}\theta}(\mathrm{e}^{\mathrm{i}\theta}I - A_{\mathrm{in}})^{-1}
    = (I - \mathrm{e}^{-\mathrm{i}\theta}A_{\mathrm{in}})^{-1}
    = \sum_{m \geq 0} A_{\mathrm{in}}^{m}\, \mathrm{e}^{-\mathrm{i}m\theta}.
\end{equation}
Multiplying by $\mathrm{e}^{\mathrm{i}k\theta}$ and integrating over $[0, 2\pi]$ selects the term $m = k$, giving $A_{\mathrm{in}}^{k}$ for $k \geq 0$ and $0$ otherwise. On $\mathcal{H}_{\mathrm{out}}$ all eigenvalues of $A_{\mathrm{out}}$ have modulus above $1$, so $\norm{A_{\mathrm{out}}^{-1}} < 1$ and
\begin{equation}
    \mathrm{e}^{\mathrm{i}\theta}(\mathrm{e}^{\mathrm{i}\theta}I - A_{\mathrm{out}})^{-1}
    = -A_{\mathrm{out}}^{-1} \frac{\mathrm{e}^{\mathrm{i}\theta}}{I - \mathrm{e}^{\mathrm{i}\theta}A_{\mathrm{out}}^{-1}}
    = -\sum_{m \geq 0} A_{\mathrm{out}}^{-(m+1)}\, \mathrm{e}^{\mathrm{i}(m+1)\theta}.
\end{equation}
Multiplying by $\mathrm{e}^{\mathrm{i}k\theta}$ and integrating selects $m + 1 = -k$, that is $m = -k-1 \geq 0$, which requires $k \leq -1$; this yields $-A_{\mathrm{out}}^{k}$ for $k \leq -1$ and $0$ otherwise.
\end{proof}

\begin{theorem}[Discrete Fourier transform of matrix resolvent]
\label{thm:app-dft-resolvent}
Let $A$ be a square matrix whose spectrum is disjoint from the unit circle, with Riesz projector $\Pi_{\mathrm{in}}$ \eqref{eq:app-riesz-unitcircle}, complementary projector $\Pi_{\mathrm{out}} = I - \Pi_{\mathrm{in}}$, and compressions $A_{\mathrm{in}} = A\Pi_{\mathrm{in}}$, $A_{\mathrm{out}} = A\Pi_{\mathrm{out}}$. For $0 \leq m \leq n-1$, the discrete Fourier transform \eqref{eq:app-mode-def} of the resolvent is given by
\begin{equation}
  \label{eq:app-closed-modes}
  S_m
  = \frac{1}{n} \sum_{j=0}^{n-1} \omega_n^{jm} (I - \omega_n^{-j} A)^{-1}
  = A_{\mathrm{in}}^{m} (I - A_{\mathrm{in}}^{n})^{-1} \Pi_{\mathrm{in}}
  - A_{\mathrm{out}}^{m-n} (I - A_{\mathrm{out}}^{-n})^{-1} \Pi_{\mathrm{out}},
\end{equation}
where $A_{\mathrm{out}}^{-1}$ denotes the inverse of $A_{\mathrm{out}}$ on $\mathcal{H}_{\mathrm{out}}$ and $\omega_n = \mathrm{e}^{2\pi\mathrm{i}/n}$.
\end{theorem}

\begin{proof}
By \lem{app-aliasing}, $S_m = \sum_{q \in \mathbb{Z}} C_{m+qn}$, and by \lem{app-continuous} the coefficients $C_k$ are nonzero only when $k \geq 0$ on $\mathcal{H}_{\mathrm{in}}$ and $k \leq -1$ on $\mathcal{H}_{\mathrm{out}}$. On $\mathcal{H}_{\mathrm{in}}$, for $0 \leq m \leq n-1$ the surviving terms are $m, m+n, m+2n, \dots$,
\begin{equation}
  S_m\big|_{\mathcal{H}_{\mathrm{in}}}
  = \sum_{q \geq 0} A_{\mathrm{in}}^{m+qn}
  = A_{\mathrm{in}}^{m} (I - A_{\mathrm{in}}^{n})^{-1}.
\end{equation}
On $\mathcal{H}_{\mathrm{out}}$ the surviving terms are $m-n, m-2n, \dots$,
\begin{equation}
  S_m\big|_{\mathcal{H}_{\mathrm{out}}}
  = -\sum_{q \geq 1} A_{\mathrm{out}}^{m-qn}
  = -A_{\mathrm{out}}^{m-n} (I - A_{\mathrm{out}}^{-n})^{-1}.
\end{equation}
Summing the two contributions gives \eqref{eq:app-closed-modes}.
\end{proof}

The closed form of \thm{app-dft-resolvent} reduces the bounds we need on the Fourier modes $S_m$ to the growth of the powers $A_{\mathrm{in}}^{j}$ and the inverse powers $A_{\mathrm{out}}^{-j}$, which decay geometrically under either of our assumptions on $A$.

\begin{corollary}[Fourier mode bounds]
\label{cor:app-mode-bounds}
Let $A$ be a square matrix with Riesz projector $\Pi_{\mathrm{in}}$ \eqref{eq:app-riesz-unitcircle} and compressions $A_{\mathrm{in}} = A\Pi_{\mathrm{in}}$, $A_{\mathrm{out}} = A\Pi_{\mathrm{out}}$, and let $S_m$ \eqref{eq:app-mode-def} be the discrete Fourier transform of the resolvent
\eqref{eq:app-closed-modes}.
\begin{enumerate}
  \item If the compressions $A_{\mathrm{in}}$ and $A_{\mathrm{out}}^{-1}$ admit spectral sets \eqref{eq:app-spectral-set} with constant $c$ and gap $\delta$ satisfying \eqref{eq:app-spectral-hyp}, then their powers decay geometrically,
    \begin{equation}
      \label{eq:app-power-decay}
      \norm{A_{\mathrm{in}}^{j}} \leq c\,(1 - \delta)^{j},
      \qquad
      \norm{A_{\mathrm{out}}^{-j}} \leq c\,(1 - \delta)^{j},
      \qquad j \geq 0.
    \end{equation}
    In particular, this holds with:
    \begin{enumerate}
      \item $c = 1+\sqrt{2}$ and $\delta = \delta_{\mathrm{num}}$ under the numerical-range
        assumption \eqref{eq:app-nr-assumption};
      \item $c = \kappa_V$ and $\delta = \delta_{\mathrm{eig}}/2$ under the diagonalizability assumption \eqref{eq:app-kappa-V}.
    \end{enumerate}
  \item Whenever \eqref{eq:app-power-decay} holds, once $n \geq \frac{1}{\delta}\log(2c)$:
    \begin{enumerate}
      \item the zeroth mode approximates the eigenprojection,
        \begin{equation}
          \label{eq:app-discretization-error}
          \norm{S_0 - \Pi_{\mathrm{in}}} \leq 4c\,(1 - \delta)^{n}\,\norm{\Pi_{\mathrm{in}}};
        \end{equation}
      \item every Fourier mode is uniformly bounded,
        \begin{equation}
          \label{eq:app-mode-uniform}
          \norm{S_m} \leq 4c\,\norm{\Pi_{\mathrm{in}}},
          \qquad 0 \leq m \leq n-1.
        \end{equation}
    \end{enumerate}
\end{enumerate}
\end{corollary}

\begin{proof}
For part (1), apply \eqref{eq:app-spectral-set} to $g(z) = z^{j}$ on $\mathcal{K}_{\mathrm{in}}$ for $A_{\mathrm{in}}$, and on $\mathcal{K}_{\mathrm{out}}$ for $A_{\mathrm{out}}^{-1}$. Both sets lie in $\{z \colon \abs{z} \leq 1 - \delta\}$, so $\norm{A_{\mathrm{in}}^{j}} \leq c(1-\delta)^{j}$ and $\norm{A_{\mathrm{out}}^{-j}} \leq c(1-\delta)^{j}$. The two special cases follow from \lem{app-spectral-instances}. In case (b), every eigenvalue of $A_{\mathrm{out}}$ has modulus at least $1 + \delta_{\mathrm{eig}}$, so every eigenvalue of $A_{\mathrm{out}}^{-1}$ has modulus at most $1/(1+\delta_{\mathrm{eig}}) \leq 1 - \delta_{\mathrm{eig}}/2$. This gives $\delta = \delta_{\mathrm{eig}}/2$.

For part (2), assume \eqref{eq:app-power-decay}. Since $n \geq \frac{1}{\delta}\log(2c)$ and $(1-\delta)^{n} \leq \mathrm{e}^{-\delta n}$,
\begin{equation}
  \norm{A_{\mathrm{in}}^{n}}, \norm{A_{\mathrm{out}}^{-n}}
  \leq c\,(1-\delta)^{n} \leq c\,\mathrm{e}^{-\delta n} \leq \tfrac{1}{2},
\end{equation}
so $(I - A_{\mathrm{in}}^{n})^{-1}$ and $(I - A_{\mathrm{out}}^{-n})^{-1}$ have norm at most $2$.

For (2a), setting $m = 0$ in the closed form \eqref{eq:app-closed-modes} and subtracting $\Pi_{\mathrm{in}} = (I - A_{\mathrm{in}}^{n})^{-1}(I - A_{\mathrm{in}}^{n})\Pi_{\mathrm{in}}$,
\begin{equation}
  S_0 - \Pi_{\mathrm{in}}
  = A_{\mathrm{in}}^{n}(I - A_{\mathrm{in}}^{n})^{-1}\Pi_{\mathrm{in}}
  - A_{\mathrm{out}}^{-n}(I - A_{\mathrm{out}}^{-n})^{-1}\Pi_{\mathrm{out}}.
\end{equation}
Taking norms with $\norm{\Pi_{\mathrm{out}}} = \norm{\Pi_{\mathrm{in}}}$ (\lem{app-projector-norm}), the inverse factors at most $2$, and $\norm{A_{\mathrm{in}}^{n}}, \norm{A_{\mathrm{out}}^{-n}} \leq c(1-\delta)^{n}$ gives \eqref{eq:app-discretization-error}.

For (2b), the closed form \eqref{eq:app-closed-modes} with $\norm{A_{\mathrm{in}}^{m}}, \norm{A_{\mathrm{out}}^{m-n}} \leq c$ for $0 \leq m \leq n-1$, the inverse factors $(I - A_{\mathrm{in}}^{n})^{-1}$ and $(I - A_{\mathrm{out}}^{-n})^{-1}$ at most $2$ in spectral norm, and $\norm{\Pi_{\mathrm{out}}} = \norm{\Pi_{\mathrm{in}}}$ gives
\begin{equation}
  \norm{S_m}
  \leq \norm{A_{\mathrm{in}}^{m}}\,\norm{(I - A_{\mathrm{in}}^{n})^{-1}}\,\norm{\Pi_{\mathrm{in}}}
  + \norm{A_{\mathrm{out}}^{m-n}}\,\norm{(I - A_{\mathrm{out}}^{-n})^{-1}}\,\norm{\Pi_{\mathrm{out}}}
  \leq 4c\,\norm{\Pi_{\mathrm{in}}}.
\end{equation}
\end{proof}

\subsection{Two-sided block preconditioning of resolvent integral}
\label{append:precond-precondition}

The standard block encoding of \append{resolvent} inverts the block-diagonal matrix $M = \sum_{j=0}^{n-1} \ketbra{j}{j} \otimes (I - \omega_n^{-j} A)$ \eqref{eq:app-M-blockdiag} and extracts the single block $S_0 \approx \Pi_{\mathrm{in}}$. However, this normalizes $\Pi_{\mathrm{in}}$ by the norm of the whole inverse matrix $\norm{M^{-1}} = \alpha_{\mathrm{res}} = \operatorname{\mathbf{O}}(\norm{\Pi_{\mathrm{in}}}/\delta)$, even though the extracted block has norm only $\operatorname{\mathbf{O}}(\norm{\Pi_{\mathrm{in}}})$ by \cor{app-mode-bounds}. The two-sided preconditioning rescales the ancilla registers so that the extracted block is normalized near $\norm{\Pi_{\mathrm{in}}}$, while the condition number of the inversion, and hence its query complexity of quantum implementation, is preserved up to a logarithmic factor.

Fix two scaling parameters $0 < t_1, t_2 \leq 1$ and define diagonal scaling operators on the ancilla
register,
\begin{equation}
  \label{eq:app-scaling-ops}
  T_1 = t_1 \ketbra{0}{0} + \sum_{k=1}^{n-1} \ketbra{k}{k}
  = \begin{bmatrix}
    t_1 & & & \\
        & 1 & & \\
        & & \ddots & \\
        & & & 1
  \end{bmatrix},
  \qquad
  T_2 = t_2 \ketbra{0}{0} + \sum_{k=1}^{n-1} \ketbra{k}{k}
  = \begin{bmatrix}
    t_2 & & & \\
        & 1 & & \\
        & & \ddots & \\
        & & & 1
  \end{bmatrix}.
\end{equation}
A block encoding of the preconditioned matrix
\begin{equation}
  \label{eq:app-precond-matrix}
  (T_1 \otimes I)(\operatorname{\mathbf{QFT}} \otimes I)\, M\, (\operatorname{\mathbf{QFT}}^\dagger \otimes I)(T_2 \otimes I)
\end{equation}
is then obtained from the block encoding of $M$ using the product rule from~\lem{app-be-product}. Note that $T_1$ and $T_2$ are diagonal contractions and the quantum Fourier transform $\mathbf{QFT}$ is a unitary; hence the preconditioned matrix can be block encoded with no normalization overhead.

Applying the inversion rule of~\lem{app-be-inversion} block encodes the inverse of the preconditioned matrix \eqref{eq:app-precond-matrix}:
\begin{equation}
  \label{eq:app-precond-inverse}
  \left[ (T_1 \otimes I)(\operatorname{\mathbf{QFT}} \otimes I)\, M\, (\operatorname{\mathbf{QFT}}^\dagger \otimes I)(T_2 \otimes I) \right]^{-1}
  = (T_2^{-1} \otimes I)(\operatorname{\mathbf{QFT}} \otimes I)\, M^{-1}\, (\operatorname{\mathbf{QFT}}^\dagger \otimes I)(T_1^{-1} \otimes I).
\end{equation}
Here, the middle factor $(\operatorname{\mathbf{QFT}} \otimes I)\, M^{-1}\, (\operatorname{\mathbf{QFT}}^\dagger \otimes I)$ is the block-circulant operator \eqref{eq:app-block-circulant}, with $(k,l)$ block $S_{k-l}$.
Taking the scaling operators $T_1$ and $T_2$ into account, the $(k,l)$ block is $S_{k-l}$ scaled by $1/(t_1 t_2)$ when $k = l = 0$, by $1/t_2$ on the rest of row $0$, by $1/t_1$ on the rest of column $0$, and left unchanged otherwise:
\begin{equation}
  \left( T_2^{-1} \otimes I \right)
  \left( \operatorname{\mathbf{QFT}} \otimes I \right)
  M^{-1}
  \left( \operatorname{\mathbf{QFT}}^\dagger \otimes I \right)
  \left( T_1^{-1} \otimes I \right)
  \;=\;
\begin{bmatrix}
  \dfrac{1}{t_1 t_2}\, S_0 & \dfrac{1}{t_2}\, S_{n-1} & \dfrac{1}{t_2}\, S_{n-2} & \cdots & \dfrac{1}{t_2}\, S_1 \\[10pt]
  \dfrac{1}{t_1}\, S_1     & S_0                      & S_{n-1}                  & \cdots & S_2 \\[10pt]
  \dfrac{1}{t_1}\, S_2     & S_1                      & S_0                      & \cdots & S_3 \\[10pt]
  \vdots                   & \vdots                   & \vdots                   & \ddots & \vdots \\[10pt]
  \dfrac{1}{t_1}\, S_{n-1} & S_{n-2}                  & S_{n-3}                  & \cdots & S_0
\end{bmatrix}.
  \label{eq:app-preconditioned-inverse}
\end{equation}

\begin{proposition}[Preconditioned norm bound]
\label{prop:app-precond-norm}
Let $0 < t_1, t_2 \leq 1$, and let $S_m$ \eqref{eq:app-mode-def} be the discrete Fourier transform of the resolvent. The inverse \eqref{eq:app-preconditioned-inverse} of the preconditioned matrix has norm bounded by
\begin{equation}
  \label{eq:app-precond-norm-bound}
\begin{aligned}
    \alpha_{\mathrm{precond}}
    &=\norm{(T_2^{-1} \otimes I)(\operatorname{\mathbf{QFT}} \otimes I) M^{-1} (\operatorname{\mathbf{QFT}}^\dagger \otimes I)(T_1^{-1} \otimes I)}\\
  &\leq \frac{\norm{S_0}}{t_1 t_2}
  + \frac{\sqrt{n-1}}{t_1}\, \max_{m} \norm{S_m}
  + \frac{\sqrt{n-1}}{t_2}\, \max_{m} \norm{S_m}
  + \norm{M^{-1}}.
\end{aligned}
\end{equation}
\end{proposition}

\begin{proof}
Applying the triangle inequality to the four blocks of \eqref{eq:app-preconditioned-inverse}---the corner, the rest of the top row, the rest of the left column, and the remaining submatrix---gives
\begin{equation}
  \label{eq:app-precond-four-terms}
  \alpha_{\mathrm{precond}}
  \leq \frac{\norm{S_0}}{t_1 t_2}
  + \frac{1}{t_2}\norm{\begin{bmatrix} S_{n-1} & \cdots & S_1 \end{bmatrix}}
  + \frac{1}{t_1} \norm{\begin{bmatrix} S_1 \\ \vdots \\ S_{n-1} \end{bmatrix}}
  + \norm{\begin{bmatrix} S_0 & \cdots & S_2 \\ \vdots & \ddots & \vdots \\ S_{n-2} & \cdots & S_0 \end{bmatrix}}.
\end{equation}
Here, the first term is $\frac{\norm{S_0}}{t_1 t_2}$. The second is a single row of $n-1$ blocks, with norm at most 
\begin{equation}
\begin{aligned}
  \norm{\begin{bmatrix} S_{n-1} & \cdots & S_1 \end{bmatrix}}
  &= \norm{\begin{bmatrix} S_{n-1} & \cdots & S_1 \end{bmatrix}
\begin{bmatrix} S_{n-1}^\dagger \\ \vdots \\ S_1^\dagger \end{bmatrix}}^{1/2}
  = \norm{\sum_{l=1}^{n-1} S_{n-l} S_{n-l}^\dagger}^{1/2}\\
  &\leq \left( \sum_{l=1}^{n-1} \norm{S_{n-l}}^2 \right)^{1/2}
  \leq \sqrt{n-1}\, \max_{m} \norm{S_m}.
\end{aligned}
\end{equation}
The third is a single column of $n-1$ blocks and can be handled similarly. Finally, the fourth is a principal submatrix of the unscaled block-circulant operator \eqref{eq:app-block-circulant}, so its norm is at most $\norm{\left( \operatorname{\mathbf{QFT}} \otimes I \right) M^{-1} \left( \operatorname{\mathbf{QFT}}^\dagger \otimes I \right)}=\norm{M^{-1}}$. Substituting these into \eqref{eq:app-precond-four-terms} gives \eqref{eq:app-precond-norm-bound}.
\end{proof}

As $t_1, t_2$ decrease, the corner and the block row and column in \eqref{eq:app-precond-norm-bound} grow, so $\alpha_{\mathrm{precond}}$ increases; but the extracted block encoding is normalized by $t_1 t_2\,\alpha_{\mathrm{precond}}$, which decreases at a higher rate. To manage this tradeoff, we take
\begin{equation}
  \label{eq:app-precond-t-choice}
  t_1 = t_2 = \frac{1}{\sqrt{n}}.
\end{equation}
Substituting into
\eqref{eq:app-precond-norm-bound},
\begin{equation}
  \label{eq:app-precond-balanced}
  \alpha_{\mathrm{precond}}
  \leq n\,\norm{S_0}
  + 2\sqrt{n(n-1)}\,\max_{m}\norm{S_m}
  + \alpha_{\mathrm{res}}
  \leq 3n\,\max_{m}\norm{S_m}
  + \alpha_{\mathrm{res}},
\end{equation}
and the extracted normalization factor becomes
\begin{equation}
  \label{eq:app-precond-extracted}
  \alpha_{\Pi}
  =t_1 t_2\,\alpha_{\mathrm{precond}}
  \leq 3\,\max_{m}\norm{S_m}
  + \frac{\alpha_{\mathrm{res}}}{n}.
\end{equation}

\begin{theorem}[Quantum oblique eigenprojection on the unit disk]
\label{thm:app-disk-main}
Let $A$ be a square matrix accessed through a block encoding $A/\alpha_A$ with normalization factor $\alpha_A$, whose spectrum is disjoint from the unit circle, with Riesz projector $\Pi_{\mathrm{in}}$ \eqref{eq:app-riesz-unitcircle} and compressions $A_{\mathrm{in}} = A\Pi_{\mathrm{in}}$, $A_{\mathrm{out}} = A\Pi_{\mathrm{out}}$.

Assume the compressions $A_{\mathrm{in}}$ and $A_{\mathrm{out}}^{-1}$ admit spectral sets \eqref{eq:app-spectral-set} with constant $c$ and gap $\delta$ satisfying \eqref{eq:app-spectral-hyp}. For any $\varepsilon > 0$, the two-sided block preconditioning \eqref{eq:app-preconditioned-inverse} with $t_1 = t_2 = 1/\sqrt{n}$ and $n = \operatorname{\pmb{\Theta}}(\frac{1}{\delta}\log\frac{c\,\norm{\Pi_{\mathrm{in}}}}{\varepsilon})$ block encodes $\Pi_{\mathrm{in}}$ with normalization factor
\begin{equation}
  \label{eq:app-main-normalization}
  \alpha_\Pi = \operatorname{\mathbf{O}}(c\,\norm{\Pi_{\mathrm{in}}}),
\end{equation}
near its optimal value $\norm{\Pi_{\mathrm{in}}}$, and accuracy $\varepsilon$, using
\begin{equation}
  \label{eq:app-main-queries}
  \operatorname{\mathbf{O}}\!\left(\frac{c\,\alpha_A\,\norm{\Pi_{\mathrm{in}}}}{\delta}
  \log\frac{c\,\norm{\Pi_{\mathrm{in}}}}{\varepsilon} \log\frac{1}{\varepsilon}\right)
\end{equation}
queries to the block encoding $A/\alpha_A$. In particular, this holds:
\begin{enumerate}
  \item with $c = 1+\sqrt{2}$ and $\delta = \delta_{\mathrm{num}}$ under the numerical-range
    assumption \eqref{eq:app-nr-assumption};
  \item with $c = \kappa_V$ and $\delta = \delta_{\mathrm{eig}}/2$ under the diagonalizability
    assumption \eqref{eq:app-kappa-V}.
\end{enumerate}
\end{theorem}

\begin{proof}
By \prop{app-precond-norm} with $t_1 = t_2 = 1/\sqrt{n}$, the extracted normalization factor is $\alpha_\Pi \leq 3\max_m\norm{S_m} + \alpha_{\mathrm{res}}/n$ \eqref{eq:app-precond-extracted}, and the norm of the inverse preconditioned matrix is $\alpha_{\mathrm{precond}} \leq 3n\max_m\norm{S_m} + \alpha_{\mathrm{res}}$ \eqref{eq:app-precond-balanced}. By part (1) of \cor{app-mode-bounds}, the spectral-set hypothesis gives the geometric decay \eqref{eq:app-power-decay}. Part (2) of \cor{app-mode-bounds} then gives $\max_m\norm{S_m} = \operatorname{\mathbf{O}}(c\,\norm{\Pi_{\mathrm{in}}})$. \prop{app-resolvent-bound} gives $\alpha_{\mathrm{res}} = \operatorname{\mathbf{O}}(c\,\norm{\Pi_{\mathrm{in}}}/\delta)$.

The choice of $n$ balances two competing effects. It must be large enough that the discretization error \eqref{eq:app-discretization-error} falls below $\varepsilon$, which requires $n = \operatorname{\pmb{\Omega}}(\frac{1}{\delta}\log\frac{c\,\norm{\Pi_{\mathrm{in}}}}{\varepsilon})$. It must also be no larger than this, since the condition number $\alpha_{\mathrm{precond}} = \operatorname{\mathbf{O}}(n\,c\,\norm{\Pi_{\mathrm{in}}}+c\,\norm{\Pi_{\mathrm{in}}}/\delta)$ grows linearly in $n$. The choice $n = \operatorname{\pmb{\Theta}}(\frac{1}{\delta}\log\frac{c\,\norm{\Pi_{\mathrm{in}}}}{\varepsilon})$ meets both. For sufficiently small $\varepsilon$, it also meets the threshold $n \geq \frac{1}{\delta}\log(2c)$ of \cor{app-mode-bounds}.

With this $n$, the discretization error is at most $\varepsilon$ by \cor{app-mode-bounds}, and $\alpha_{\mathrm{res}}/n = \operatorname{\mathbf{O}}(c\,\norm{\Pi_{\mathrm{in}}})$. So $\alpha_\Pi = \operatorname{\mathbf{O}}(c\,\norm{\Pi_{\mathrm{in}}})$, which is \eqref{eq:app-main-normalization}. Also, $\alpha_{\mathrm{precond}} = \operatorname{\mathbf{O}}(n\,c\,\norm{\Pi_{\mathrm{in}}}) = \operatorname{\mathbf{O}}(\frac{c\,\norm{\Pi_{\mathrm{in}}}}{\delta}\log\frac{c\,\norm{\Pi_{\mathrm{in}}}}{\varepsilon})$. By \eqref{eq:app-precond-matrix}, the preconditioned matrix is block encoded at normalization $\operatorname{\mathbf{O}}(\alpha_A)$. Inverting it with \lem{app-be-inversion} uses $\operatorname{\mathbf{O}}(\alpha_A\,\alpha_{\mathrm{precond}}\log\frac{1}{\varepsilon})$ queries, which is \eqref{eq:app-main-queries}. Extracting the top-left block gives the block encoding of $\Pi_{\mathrm{in}}$ at normalization $\alpha_\Pi$. The two special cases follow from \lem{app-spectral-instances}, with $\delta = \delta_{\mathrm{eig}}/2$ in case (ii) as in \cor{app-mode-bounds}.
\end{proof}

%%%%%%%%%%%%%%%%%%%%%%%%%%%%%%%%%%%%%%%%%%%%%%%%%%%%%%%%%%%%%%%%%%%%%%%%%%%%%%
\section{Extension to general regions}
\label{append:region}

In this appendix, we provide details on the technical results from~\sec{regions}.

\subsection{Reduction to the unit disk}
\label{append:region-reduction}

\append{precond} treats the unit disk. We now reduce a general region to that case through a separating map.

Let $A$ be a square matrix with spectrum split into $\operatorname{\mathbf{Spec}}_{\mathrm{in}}(A)$ and $\operatorname{\mathbf{Spec}}_{\mathrm{out}}(A)$, and $\mathcal{C}$ be a closed curve enclosing $\operatorname{\mathbf{Spec}}_{\mathrm{in}}(A)$ but not $\operatorname{\mathbf{Spec}}_{\mathrm{out}}(A)$, disjoint from $\operatorname{\mathbf{Spec}}(A)$. Then, the oblique eigenprojection onto $\operatorname{\mathbf{Spec}}_{\mathrm{in}}(A)$ is given by the Riesz projector $\Pi_{\mathrm{in}} = \frac{1}{2\pi\mathrm{i}} \int_{\mathcal{C}} (w I - A)^{-1}\, \mathrm{d}w$. 

Suppose that $\mathcal{S}_{\mathrm{in}}$ and $\mathcal{S}_{\mathrm{out}}$ are disjoint open sets covering $\operatorname{\mathbf{Spec}}_{\mathrm{in}}(A)$ and $\operatorname{\mathbf{Spec}}_{\mathrm{out}}(A)$. Let $f$ be analytic on $\mathcal{S}_{\mathrm{in}} \cup \mathcal{S}_{\mathrm{out}}$ and separate the two sets in modulus,
\begin{equation}
  \label{eq:app-separation}
  \abs{f(w)} < 1 \quad (w \in \mathcal{S}_{\mathrm{in}}),
  \qquad
  \abs{f(w)} > 1 \quad (w \in \mathcal{S}_{\mathrm{out}}).
\end{equation}
The next lemma states that the eigenprojection of $A$ for the target region coincides with the eigenprojection of $f(A)$ for the unit disk.

\begin{lemma}[Change of variables]
\label{lem:app-cov}
Let $A$ be a square matrix with spectrum split into $\operatorname{\mathbf{Spec}}_{\mathrm{in}}(A)$ and $\operatorname{\mathbf{Spec}}_{\mathrm{out}}(A)$, and $\mathcal{C}$ be a closed curve enclosing $\operatorname{\mathbf{Spec}}_{\mathrm{in}}(A)$ but not $\operatorname{\mathbf{Spec}}_{\mathrm{out}}(A)$, disjoint from $\operatorname{\mathbf{Spec}}(A)$. 
Suppose that $\mathcal{S}_{\mathrm{in}}$ and $\mathcal{S}_{\mathrm{out}}$ are disjoint open sets covering $\operatorname{\mathbf{Spec}}_{\mathrm{in}}(A)$ and $\operatorname{\mathbf{Spec}}_{\mathrm{out}}(A)$. Let $f$ be analytic on $\mathcal{S}_{\mathrm{in}} \cup \mathcal{S}_{\mathrm{out}}$ with $\abs{f(w)} < 1 \ (w \in \mathcal{S}_{\mathrm{in}})$ and $\abs{f(w)} > 1 \ (w \in \mathcal{S}_{\mathrm{out}})$. Then,
\begin{equation}
  \label{eq:app-cov}
  \Pi_{\mathrm{in}}
  = \frac{1}{2\pi \mathrm{i}} \int_{\mathcal{C}} (w I - A)^{-1} \, \mathrm{d}w
  = \frac{1}{2\pi \mathrm{i}} \int_{\abs{z}=1} (z I - f(A))^{-1} \, \mathrm{d}z.
\end{equation}
\end{lemma}

\begin{proof}
Both integrals are Riesz projectors. By \lem{app-riesz-image}, the first is the projection onto the invariant subspaces of $A$ associated with $\operatorname{\mathbf{Spec}}_{\mathrm{in}}(A)$, and the second is the projection onto the invariant subspaces of $f(A)$ associated with the part of its spectrum inside the unit circle. By the spectral mapping theorem~\cite[Theorem 12.7.1]{humpherys2017}\cite[Theorem 8.3]{roman2008advanced}, $\operatorname{\mathbf{Spec}}(f(A)) = f(\operatorname{\mathbf{Spec}}(A))$. The modulus separation map $f$ sends $\operatorname{\mathbf{Spec}}_{\mathrm{in}}(A)$ inside the unit circle and $\operatorname{\mathbf{Spec}}_{\mathrm{out}}(A)$ outside, so the two projections have matching spectral parts. The holomorphic functional calculus preserves the invariant subspaces of $A$, so $f(A)$ acts on each invariant subspace of $A$ and shares its spectral projection onto $\operatorname{\mathbf{Spec}}_{\mathrm{in}}(A)$. The two integrals therefore coincide.
\end{proof}

\begin{theorem}[Quantum oblique eigenprojection for general regions]
\label{thm:app-region-main}
Let $A$ be a square matrix with Riesz projector $\Pi_{\mathrm{in}}$ \eqref{eq:app-cov} for a region whose boundary is the separating curve $\mathcal{C}$, and compressions $A_{\mathrm{in}} = A\Pi_{\mathrm{in}}$, $A_{\mathrm{out}} = A\Pi_{\mathrm{out}}$. Let $f$ be an analytic separating map \eqref{eq:app-separation}, and suppose $f(A)$ is accessed through a block encoding $f(A)/\alpha_{f(A)}$ with normalization factor $\alpha_{f(A)}$.

If the compressions $A_{\mathrm{in}}$ and $A_{\mathrm{out}}$ admit spectral sets $\mathcal{K}_{\mathrm{in}}$ and $\mathcal{K}_{\mathrm{out}}$ \eqref{eq:app-spectral-set} with constant $c$, and $f$ is analytic on neighborhoods of $\mathcal{K}_{\mathrm{in}}$ and $\mathcal{K}_{\mathrm{out}}$ with gap $\delta_f$ satisfying
\begin{equation}
  \label{eq:app-region-gap}
  \max_{z \in \mathcal{K}_{\mathrm{in}}} \abs{f(z)} \leq 1 - \delta_f,
  \qquad
  \min_{z \in \mathcal{K}_{\mathrm{out}}} \abs{f(z)} \geq 1 + \delta_f,
\end{equation}
then for any $\varepsilon > 0$, the two-sided block preconditioning of \thm{app-disk-main} applied to $f(A)$ block encodes $\Pi_{\mathrm{in}}$ with normalization factor
\begin{equation}
  \label{eq:app-region-normalization}
  \alpha_\Pi = \operatorname{\mathbf{O}}(c\,\norm{\Pi_{\mathrm{in}}}),
\end{equation}
near its optimal value $\norm{\Pi_{\mathrm{in}}}$, and accuracy $\varepsilon$, using
\begin{equation}
  \label{eq:app-region-queries}
  \operatorname{\mathbf{O}}\!\left(\frac{c\,\alpha_{f(A)}\,\norm{\Pi_{\mathrm{in}}}}{\delta_f}
  \log\frac{c\,\norm{\Pi_{\mathrm{in}}}}{\varepsilon} \log\frac{1}{\varepsilon}\right)
\end{equation}
queries to the block encoding $f(A)/\alpha_{f(A)}$. In particular, this holds:
\begin{enumerate}
  \item with $c = 1+\sqrt{2}$, $\mathcal{K}_{\mathrm{in}} = \mathcal{W}(A_{\mathrm{in}})$, and
    $\mathcal{K}_{\mathrm{out}} = \mathcal{W}(A_{\mathrm{out}})$, under the numerical-range
    assumption;
  \item with $c = \kappa_V$ \eqref{eq:app-kappa-V}, $\mathcal{K}_{\mathrm{in}} = \operatorname{\mathbf{Spec}}(A_{\mathrm{in}})$, and
    $\mathcal{K}_{\mathrm{out}} = \operatorname{\mathbf{Spec}}(A_{\mathrm{out}})$, under the diagonalizability
    assumption.
\end{enumerate}
\end{theorem}

\begin{proof}
By \lem{app-cov}, $\Pi_{\mathrm{in}}$ is the unit-disk Riesz projector of $f(A)$. We apply \thm{app-disk-main} to $f(A)$, with $f(A)/\alpha_{f(A)}$ as the input block encoding. It remains to show that the compressions $f(A)_{\mathrm{in}}$ and $f(A)_{\mathrm{out}}^{-1}$ admit spectral sets \eqref{eq:app-spectral-set} with constant $c$ and gap $\delta_f/2$ satisfying \eqref{eq:app-spectral-hyp}.

Since $f$ preserves the invariant subspaces, $f(A)_{\mathrm{in}}$ acts on $\mathcal{H}_{\mathrm{in}}$ as $f(A_{\mathrm{in}})$. Similarly, $f(A)_{\mathrm{out}}^{-1}$ acts on $\mathcal{H}_{\mathrm{out}}$ as $(1/f)(A_{\mathrm{out}})$. The function $1/f$ is analytic on a neighborhood of $\mathcal{K}_{\mathrm{out}}$, because $\abs{f} \geq 1 + \delta_f > 0$ there.

We claim that $f(\mathcal{K}_{\mathrm{in}})$ is a spectral set for $f(A_{\mathrm{in}})$ with constant $c$. The set $f(\mathcal{K}_{\mathrm{in}})$ is compact, and it contains $\operatorname{\mathbf{Spec}}(f(A_{\mathrm{in}})) = f(\operatorname{\mathbf{Spec}}(A_{\mathrm{in}}))$ by the spectral mapping theorem. Let $g$ be analytic on a neighborhood of $f(\mathcal{K}_{\mathrm{in}})$. Then $g \circ f$ is analytic on a neighborhood of $\mathcal{K}_{\mathrm{in}}$, and $g(f(A_{\mathrm{in}})) = (g\circ f)(A_{\mathrm{in}})$. Applying \eqref{eq:app-spectral-set} to $g \circ f$ gives
\begin{equation}
  \norm{g(f(A_{\mathrm{in}}))}
  \leq c \sup_{z \in \mathcal{K}_{\mathrm{in}}} \abs{g(f(z))}
  = c \sup_{w \in f(\mathcal{K}_{\mathrm{in}})} \abs{g(w)}.
\end{equation}
The same argument with $1/f$ in place of $f$ shows that $(1/f)(\mathcal{K}_{\mathrm{out}})$ is a spectral set for $f(A)_{\mathrm{out}}^{-1}$ with constant $c$.

By \eqref{eq:app-region-gap}, $f(\mathcal{K}_{\mathrm{in}})$ lies in $\{\abs{w} \leq 1 - \delta_f\}$. The out-block gap converts to a reciprocal bound: $\abs{1/f} \leq \frac{1}{1+\delta_f} \leq 1 - \frac{\delta_f}{2}$ on $\mathcal{K}_{\mathrm{out}}$. So $(1/f)(\mathcal{K}_{\mathrm{out}})$ lies in $\{\abs{w} \leq 1 - \delta_f/2\}$. Both sets lie in $\{\abs{w} \leq 1 - \delta_f/2\}$, so \eqref{eq:app-spectral-hyp} holds with gap $\delta_f/2$. \thm{app-disk-main} with $\delta = \delta_f/2 = \operatorname{\pmb{\Theta}}(\delta_f)$ then gives the normalization factor \eqref{eq:app-region-normalization} and the query count \eqref{eq:app-region-queries}. The two special cases follow from \lem{app-spectral-instances}.
\end{proof}

When $f(z) = z$, \thm{app-region-main} reduces to the unit-disk setting, with the out-block hypothesis placed on $A_{\mathrm{out}}$ rather than $A_{\mathrm{out}}^{-1}$. In the diagonalizable case, taking $\mathcal{K}_{\mathrm{out}} = \operatorname{\mathbf{Spec}}(A_{\mathrm{out}})$ reproduces case (ii) of \thm{app-disk-main}. In the numerical-range case, taking $\mathcal{K}_{\mathrm{out}} = \mathcal{W}(A_{\mathrm{out}})$ gives a different condition from \eqref{eq:app-nr-assumption}, and in general neither condition implies the other. Case (i) of \thm{app-disk-main} corresponds instead to the reciprocal set $\mathcal{K}_{\mathrm{out}} = 1/\mathcal{W}(A_{\mathrm{out}}^{-1})$. This set is not the numerical range of $A_{\mathrm{out}}$. But it is still a spectral set for $A_{\mathrm{out}}$ with constant $1+\sqrt{2}$, since the Crouzeix--Palencia bound applies to $g(1/z)$ over $\mathcal{W}(A_{\mathrm{out}}^{-1})$.

More generally, the set $f(\mathcal{W}(A_{\mathrm{in}}))$ is not the numerical range of $f(A_{\mathrm{in}})$, and need not even be convex. It is nevertheless a spectral set for $f(A_{\mathrm{in}})$ with constant $1+\sqrt{2}$, since the Crouzeix--Palencia bound applies to $g \circ f$ over $\mathcal{W}(A_{\mathrm{in}})$. The numerical range $\mathcal{W}(f(A_{\mathrm{in}}))$ may be larger and may even meet the unit circle. This is why we state and apply \thm{app-disk-main} for spectral sets rather than numerical ranges in the proof of~\thm{app-region-main}.

\subsection{Half-plane construction}
\label{append:region-halfplane}

We now focus on the case where the spectrum is split along the imaginary axis, which underlies the applications of \sec{app}. We instantiate \thm{app-region-main} with a degree-two rational map sending the right half-plane into the unit disk. Take
\begin{equation}
  \label{eq:app-halfplane-map}
  f(z) = \frac{(z - 1)^2}{(z + 1)^2},
\end{equation}
with a double zero at $z = 1$ and a double pole at $z = -1$. The map sends the right half-plane into the unit disk, the imaginary axis onto the unit circle, and the left half-plane outside. Moreover, the following lemma shows that a spectrum separation between half-planes transfers to the unit circle.

\begin{lemma}[Gap preservation]
\label{lem:app-halfplane-gap}
Let $f$ be the map $f(z) = \frac{(z - 1)^2}{(z + 1)^2}$, and write $z = x + iy$. Then
\begin{equation}
  \label{eq:app-halfplane-modulus}
  \abs{f(z)} = \frac{(x-1)^2 + y^2}{(x+1)^2 + y^2} = 1 - \frac{4x}{(x+1)^2 + y^2},
\end{equation}
so $\abs{f(z)} < 1$ for $x > 0$, $\abs{f(z)} = 1$ for $x = 0$, and $\abs{f(z)} > 1$ for $x < 0$.
Moreover, for $\abs{z} \leq 1/2$:
\begin{enumerate}
  \item $\abs{f(z)} \leq 1 - \tfrac{16}{9}\delta$ if $x \geq \delta$;
  \item $\abs{f(z)} \geq 1 + \tfrac{16}{9}\delta$ if $x \leq -\delta$.
\end{enumerate}
\end{lemma}

\begin{proof}
Since $f(z) = \left(\frac{z-1}{z+1}\right)^2$, its modulus is the squared Möbius modulus $\abs{f(z)} = \left|\frac{z-1}{z+1}\right|^2 = \frac{(x-1)^2 + y^2}{(x+1)^2 + y^2}$, which gives the explicit formula in \eqref{eq:app-halfplane-modulus}. 

The sign of $4x$ determines
whether $\abs{f(z)}$ is below, equal to, or above $1$.
For $\abs{z} \leq 1/2$ the denominator $(x+1)^2 + y^2$ is the squared distance from $-1$ to $z$, maximized at $z = 1/2$, so it is at most $(3/2)^2 = 9/4$. Hence for $x \geq \delta$,
\begin{equation}
  \abs{f(z)} = 1 - \frac{4x}{(x+1)^2 + y^2} \leq 1 - \frac{4\delta}{9/4} = 1 - \tfrac{16}{9}\delta.
\end{equation}
The case $x \leq -\delta$ is symmetric, giving $\abs{f(z)} \geq 1 + \tfrac{16}{9}\delta$.
\end{proof}

The map $f(A)$ is applied without polynomial approximation. Under $\norm{A} \leq 1/2$, the shifted matrix $A + I$ has singular values in $[1/2, 3/2]$, since $\norm{A+I}\leq\frac{3}{2}$ and $\norm{(A+I)^{-1}} \le (1 - \norm{A})^{-1} \leq 2$. Inverting it with \lem{app-be-inversion} block encodes $(A + I)^{-1}$ at normalization $\operatorname{\mathbf{O}}(1)$, using $\operatorname{\mathbf{O}}(\log(1/\varepsilon))$ queries to the block encoding of $A$. The product $f(A) = (A - I)^2 (A + I)^{-2}$ then follows from \lem{app-be-product}, with
\begin{equation}
  \norm{f(A)} \leq \norm{A - I}^2 \norm{(A+I)^{-1}}^2 \leq (3/2)^2 \, (2)^2 = 9.
\end{equation}
So $f(A)$ is block encoded at normalization $\alpha_{f(A)} = \operatorname{\mathbf{O}}(1)$ using $\operatorname{\mathbf{O}}(\log(1/\varepsilon))$ queries.

For the half-plane, the spectral-set hypothesis of \thm{app-region-main} takes a concrete geometric form. Suppose the compressions $A_{\mathrm{in}}$ and $A_{\mathrm{out}}$ admit spectral sets $\mathcal{K}_{\mathrm{in}}$ and $\mathcal{K}_{\mathrm{out}}$ \eqref{eq:app-spectral-set} with constant $c$, separated from the imaginary axis into opposite half-planes,
\begin{equation}
  \label{eq:app-halfplane-hyp}
  \mathcal{K}_{\mathrm{in}} \subseteq \{\, z \colon \Re(z)\geq \delta,\ \abs{z} \leq \tfrac{1}{2} \,\},
  \qquad
  \mathcal{K}_{\mathrm{out}} \subseteq \{\, z \colon \Re(z)\leq -\delta,\ \abs{z} \leq \tfrac{1}{2} \,\}.
\end{equation}
Since \lem{app-halfplane-gap} bounds $\abs{f(z)}$ pointwise, it applies to every point of these sets. This gives $\max_{z \in \mathcal{K}_{\mathrm{in}}}\abs{f(z)} \leq 1 - \frac{16}{9}\delta$ and $\min_{z \in \mathcal{K}_{\mathrm{out}}}\abs{f(z)} \geq 1 + \frac{16}{9}\delta$. The only pole of $f$ is at $z = -1$, which lies outside the disk $\{\abs{z} \leq \frac{1}{2}\}$. So $f$ is analytic on a neighborhood of each set. The hypothesis of \thm{app-region-main} therefore holds with $\delta_f = \frac{16}{9}\delta$.

In particular, \eqref{eq:app-halfplane-hyp} holds with $c = 1+\sqrt{2}$ when the per-block numerical ranges $\mathcal{W}(A_{\mathrm{in}})$ and $\mathcal{W}(A_{\mathrm{out}})$ satisfy it. It holds with $c = \kappa_V$ when $A_{\mathrm{in}}$ and $A_{\mathrm{out}}$ are diagonalizable and their spectra $\operatorname{\mathbf{Spec}}(A_{\mathrm{in}})$ and $\operatorname{\mathbf{Spec}}(A_{\mathrm{out}})$ satisfy it. In both cases, the condition $\abs{z} \leq \frac{1}{2}$ is automatic under $\norm{A} \leq \frac{1}{2}$.

\begin{corollary}[Quantum oblique eigenprojection for the half-plane]
\label{cor:app-halfplane}
Let $A$ be a square matrix accessed through a block encoding $A/\alpha_A$ with normalization factor $\alpha_A$. Let $\Pi_{\mathrm{in}}$ be the Riesz projector onto the right-half-plane invariant subspace, with compressions $A_{\mathrm{in}} = A\Pi_{\mathrm{in}}$, $A_{\mathrm{out}} = A\Pi_{\mathrm{out}}$.

If the compressions $A_{\mathrm{in}}$ and $A_{\mathrm{out}}$ admit spectral sets $\mathcal{K}_{\mathrm{in}}$ and $\mathcal{K}_{\mathrm{out}}$ \eqref{eq:app-spectral-set} with constant $c$ and gap $\delta$ satisfying
\begin{equation}
  \label{eq:app-halfplane-cor-hyp}
  \mathcal{K}_{\mathrm{in}} \subseteq \{\, z \colon \Re(z)\geq \delta,\ \abs{z} \leq \alpha_A \,\},
  \qquad
  \mathcal{K}_{\mathrm{out}} \subseteq \{\, z \colon \Re(z)\leq -\delta,\ \abs{z} \leq \alpha_A \,\},
\end{equation}
then for any $\varepsilon > 0$, the eigenprojection $\Pi_{\mathrm{in}}$ can be block encoded with normalization factor
\begin{equation}
  \label{eq:app-halfplane-normalization}
  \alpha_\Pi = \operatorname{\mathbf{O}}(c\,\norm{\Pi_{\mathrm{in}}}),
\end{equation}
near its optimal value $\norm{\Pi_{\mathrm{in}}}$, and accuracy $\varepsilon$, using
\begin{equation}
  \label{eq:app-halfplane-queries}
  \operatorname{\mathbf{O}}\!\left(\frac{c\,\alpha_A\,\norm{\Pi_{\mathrm{in}}}}{\delta}
  \polylog\frac{c\,\alpha_A\,\norm{\Pi_{\mathrm{in}}}}{\delta\varepsilon}\right)
\end{equation}
queries to the block encoding $A/\alpha_A$. In particular, this holds:
\begin{enumerate}
  \item with $c = 1+\sqrt{2}$, when the per-block numerical ranges
    $\mathcal{W}(A_{\mathrm{in}})$ and $\mathcal{W}(A_{\mathrm{out}})$ satisfy
    \eqref{eq:app-halfplane-cor-hyp};
  \item with $c = \kappa_V$ \eqref{eq:app-kappa-V}, when $A_{\mathrm{in}}$ and $A_{\mathrm{out}}$
    are diagonalizable and their spectra $\operatorname{\mathbf{Spec}}(A_{\mathrm{in}})$ and $\operatorname{\mathbf{Spec}}(A_{\mathrm{out}})$
    satisfy \eqref{eq:app-halfplane-cor-hyp}.
\end{enumerate}
In both cases, the condition $\abs{z} \leq \alpha_A$ is automatic, since $\norm{A} \leq \alpha_A$.
\end{corollary}

\begin{proof}
Scaling the block encoding $A/\alpha_A$ by $1/2$ block encodes $A_2 = A/(2\alpha_A)$ at normalization $\operatorname{\mathbf{O}}(1)$, with $\norm{A_2} \leq 1/2$. Scaling by the positive real $1/(2\alpha_A)$ fixes the imaginary axis. So $A_2$ has the same invariant subspaces as $A$, the same right- and left-half-plane split, and the same projector $\Pi_{\mathrm{in}}$. Spectral sets rescale with the same constant: if $\mathcal{K}$ is a spectral set for $B$ with constant $c$, then $\mathcal{K}/s$ is a spectral set for $B/s$ with constant $c$, by applying \eqref{eq:app-spectral-set} to $g(z/s)$. So $\mathcal{K}_{\mathrm{in}}/(2\alpha_A)$ and $\mathcal{K}_{\mathrm{out}}/(2\alpha_A)$ are spectral sets for the compressions of $A_2$ with constant $c$. By \eqref{eq:app-halfplane-cor-hyp}, they satisfy \eqref{eq:app-halfplane-hyp} with gap $\delta' = \delta/(2\alpha_A)$.

Take $f$ to be the map \eqref{eq:app-halfplane-map}. Under $\norm{A_2} \leq 1/2$, the product $f(A_2) = (A_2 - I)^2 (A_2 + I)^{-2}$ is block encoded at normalization $\operatorname{\mathbf{O}}(1)$ using $\operatorname{\mathbf{O}}(\log(1/\varepsilon))$ queries to the block encoding of $A_2$. By the discussion following \eqref{eq:app-halfplane-hyp}, the hypothesis of \thm{app-region-main} holds for $A_2$ with $\delta_f = \frac{16}{9}\delta' = \operatorname{\pmb{\Theta}}(\delta/\alpha_A)$. Applying \thm{app-region-main} to $f(A_2)$ gives the normalization factor \eqref{eq:app-halfplane-normalization}, using
\begin{equation}
  q_f = \operatorname{\mathbf{O}}\!\left(\frac{c\,\alpha_A\,\norm{\Pi_{\mathrm{in}}}}{\delta}
  \log\frac{c\,\norm{\Pi_{\mathrm{in}}}}{\varepsilon}\log\frac{1}{\varepsilon}\right)
\end{equation}
queries to the block encoding of $f(A_2)$.

Each of the $q_f$ queries to $f(A_2)$ is realized by an $\varepsilon'$-accurate block encoding. Since block-encoding errors compose linearly, taking $\varepsilon' = \operatorname{\pmb{\Theta}}(\varepsilon/q_f)$ gives final accuracy $\varepsilon$. This uses $\operatorname{\mathbf{O}}(\log(q_f/\varepsilon))$ queries to $A/\alpha_A$ per query to $f(A_2)$. Since $q_f$ is polynomial in the parameters, this adds only a polylogarithmic factor, which is included in \eqref{eq:app-halfplane-queries}.
\end{proof}

%%%%%%%%%%%%%%%%%%%%%%%%%%%%%%%%%%%%%%%%%%%%%%%%%%%%%%%%%%%%%%%%%%%%%%%%%%%%%%
\section{Further analysis on applications}
\label{append:app}

In this appendix, we provide details on the technical results from~\sec{app}.

\subsection{Eigenstate preparation}
\label{append:app-eigenstate}

Suppose the target eigenvalue or cluster is enclosed by a disk $\{z:\,\abs{z - z_0} < \rho\,\}$, with the target invariant subspace separated from the circle $\abs{z - z_0} = \rho$ by an absolute gap $\delta$. The separating map is the linear $f(z) = (z - z_0)/\rho$, so $f(A) = (A - z_0 I)/\rho$ is obtained from the block encoding of $A$ by a shift and a rescaling.

\begin{corollary}[Eigenstate preparation]
\label{cor:app-eigenstate}
Let $A$ be a square matrix accessed through a block encoding $A/\alpha_A$, and let $\Pi_{\mathrm{in}}$ be the Riesz projector onto the invariant subspace whose eigenvalues lie in the disk $\{\,z \colon \abs{z - z_0} < \rho\,\}$ of center $z_0$ and radius $\rho$, with compressions $A_{\mathrm{in}} = A\Pi_{\mathrm{in}}$, $A_{\mathrm{out}} = A\Pi_{\mathrm{out}}$. Let $\ket{\psi}$ be an initial state prepared by a unitary $O_\psi$, with overlap $\eta = \norm{\Pi_{\mathrm{in}}\ket{\psi}}$.

If the compressions $A_{\mathrm{in}} - z_0 I$ and $(A_{\mathrm{out}} - z_0 I)^{-1}$ admit spectral sets $\mathcal{K}_{\mathrm{in}}$ and $\mathcal{K}_{\mathrm{out}}$ \eqref{eq:app-spectral-set} with constant $c$ and gap $0 < \delta < \rho$ satisfying
\begin{equation}
  \label{eq:app-eigenstate-hyp}
  \mathcal{K}_{\mathrm{in}} \subseteq \{\, z \colon \abs{z} \leq \rho - \delta \,\},
  \qquad
  \mathcal{K}_{\mathrm{out}} \subseteq \Bigl\{\, z \colon  \abs{z} \leq \frac{1}{\rho + \delta} \,\Bigr\},
\end{equation}
then the projected state $\Pi_{\mathrm{in}}\ket{\psi}/\eta$ can be prepared to accuracy $\varepsilon$ using
\begin{equation}
  \label{eq:app-eigenstate-queries}
  \operatorname{\mathbf{O}}\!\left(\frac{c^2\,(\alpha_A + \abs{z_0})\,\norm{\Pi_{\mathrm{in}}}^2}{\eta\,\delta}
  \polylog\frac{c\,\norm{\Pi_{\mathrm{in}}}}{\eta\,\varepsilon}\right)
\end{equation}
queries to the block encoding $A/\alpha_A$. The initial state is prepared 
\begin{equation}
  \label{eq:app-eigenstate-state-prep}
  \operatorname{\mathbf{O}}\!\left(\frac{c\,\norm{\Pi_{\mathrm{in}}}}{\eta}\right)
\end{equation}
times, using that many queries to $O_\psi$ and its inverse. In particular,
\begin{enumerate}
  \item with $c = 1+\sqrt{2}$ when the numerical ranges
    $\mathcal{W}(A_{\mathrm{in}} - z_0 I)$ and $\mathcal{W}((A_{\mathrm{out}} - z_0 I)^{-1})$
    satisfy \eqref{eq:app-eigenstate-hyp};
  \item with $c = \kappa_V$ \eqref{eq:app-kappa-V} when $A_{\mathrm{in}}$ and $A_{\mathrm{out}}$ are
    diagonalizable and the spectra $\operatorname{\mathbf{Spec}}(A_{\mathrm{in}} - z_0 I)$ and
    $\operatorname{\mathbf{Spec}}((A_{\mathrm{out}} - z_0 I)^{-1})$ satisfy
    \eqref{eq:app-eigenstate-hyp}.
\end{enumerate}
\end{corollary}

\begin{proof}
Form a block encoding of $f(A) = (A - z_0 I)/\rho$ from the block encoding of $A$ by \lem{app-be-lincomb}. Its normalization factor is $\alpha_{f(A)} = (\alpha_A + \abs{z_0})/\rho$. The block encoding is exact and uses one query to the block-encoding of $A$. $\Pi_{\mathrm{in}}$ is the unit-disk Riesz projector of $f(A)$. When rescaling by $\rho$, the spectral sets also rescale by $\rho$, so $\mathcal{K}_{\mathrm{in}}/\rho$ and $\rho\,\mathcal{K}_{\mathrm{out}}$ are spectral sets for $f(A)_{\mathrm{in}}$ and $f(A)_{\mathrm{out}}^{-1} = \rho\,(A_{\mathrm{out}} - z_0 I)^{-1}$, respectively, both with constant $c$. By \eqref{eq:app-eigenstate-hyp}, the first set lies in $\{z \colon \abs{z} \leq 1 - \delta/\rho\}$. The second lies in $\{z \colon \abs{z} \leq \rho/(\rho+\delta)\}$, and $\rho/(\rho+\delta) \leq 1 - \delta/(2\rho)$. So \eqref{eq:app-spectral-hyp} holds for $f(A)$ with gap $\delta/(2\rho)$.

\thm{app-disk-main} applied to $f(A)$ block encodes $\Pi_{\mathrm{in}}$ with normalization factor $\alpha_\Pi = \operatorname{\mathbf{O}}(c\,\norm{\Pi_{\mathrm{in}}})$. It uses
\begin{equation}
  \operatorname{\mathbf{O}}\!\left(\frac{c\,\alpha_{f(A)}\,\norm{\Pi_{\mathrm{in}}}}{\delta/(2\rho)}
  \log\frac{c\,\norm{\Pi_{\mathrm{in}}}}{\varepsilon'}\log\frac{1}{\varepsilon'}\right)
  = \operatorname{\mathbf{O}}\!\left(\frac{c\,(\alpha_A + \abs{z_0})\,\norm{\Pi_{\mathrm{in}}}}{\delta}
  \log\frac{c\,\norm{\Pi_{\mathrm{in}}}}{\varepsilon'}\log\frac{1}{\varepsilon'}\right)
\end{equation}
queries per application to accuracy $\varepsilon'$. The factor $\rho$ cancels between $\alpha_{f(A)}$ and the relative gap $\delta/(2\rho)$.

Applying this block encoding to $\ket{\psi}$ produces $\Pi_{\mathrm{in}}\ket{\psi}$ with success amplitude $\eta/\alpha_\Pi$. So fixed-point amplitude amplification prepares the normalized state to constant success amplitude with $\operatorname{\mathbf{O}}(\alpha_\Pi/\eta) = \operatorname{\mathbf{O}}(c\,\norm{\Pi_{\mathrm{in}}}/\eta)$ applications. Taking each application to accuracy $\varepsilon' = \operatorname{\pmb{\Theta}}(\varepsilon\eta/\alpha_\Pi)$ gives final accuracy $\varepsilon$. Then each logarithm is $\operatorname{\mathbf{O}}(\log\frac{c\,\norm{\Pi_{\mathrm{in}}}}{\eta\,\varepsilon})$. Multiplying the number of applications by the queries per application gives \eqref{eq:app-eigenstate-queries}. The two special cases follow from \lem{app-spectral-instances}.
\end{proof}

Alternatively, the resolvent-integral construction of \append{resolvent} (and~\cite{RZL2026}) applied to $f(A)$ block encodes the same projector with normalization factor $\alpha_{\mathrm{res}} = \operatorname{\mathbf{O}}(\rho\,\norm{\Pi_{\mathrm{in}}}/\delta)$ by \prop{app-resolvent-bound}. Each application has the same query cost as before. Amplitude amplification then needs $\operatorname{\mathbf{O}}(\alpha_{\mathrm{res}}/\eta)$ rounds instead of $\operatorname{\mathbf{O}}(\alpha_\Pi/\eta)$, a factor of order $\rho/\delta$ more. For $\rho = \operatorname{\pmb{\Theta}}(1)$, the total gap dependence of queries to block encoding is $1/\delta^2$ for the resolvent integral versus $1/\delta$ with preconditioning. The latter matches the gap dependence of eigenstate preparation for nonnormal matrices with real spectra~\cite{QEVP}. Separately, the initial state is prepared once per round. So the resolvent integral uses $\operatorname{\mathbf{O}}(\rho\,\norm{\Pi_{\mathrm{in}}}/(\eta\,\delta))$ queries to $O_\psi$ and its inverse, compared with our gap-independent count \eqref{eq:app-eigenstate-state-prep}.

\subsection{Continuous-time algebraic Riccati equation}
\label{append:app-riccati}

Consider the continuous-time algebraic Riccati equation $XQX - XP - P^\dagger X - R = 0$, where $P, Q, R$ are square matrices of the same size with $Q, R$ Hermitian. Its associated Hamiltonian matrix is $J = \left[\begin{smallmatrix} P & -Q \\ -R & -P^\dagger \end{smallmatrix}\right]$. The spectrum of $J$ splits across the imaginary axis. Let $\Pi_{\mathrm{in}}$ be the projector onto the antistable (right-half-plane) invariant subspace, and write it in column blocks as $\Pi_{\mathrm{in}} = \left[\begin{smallmatrix} \Pi_1 & \Pi_2 \end{smallmatrix}\right]$. Then the stabilizing solution is $X = -\Pi_2^{+}\Pi_1$~\cite{RZL2026}, where $(\cdot)^{+}$ is the pseudoinverse. Here, the smallest nonzero singular value of $\Pi_2$ satisfies $\operatorname{\mathbf{Sval}}_{\min}(\Pi_2) \geq 1/\sqrt{1 + \norm{X}^2}$~\cite{RZL2026}. This bound does not depend on the gap.

\begin{corollary}[Riccati equation]
\label{cor:app-riccati}
Let $P, Q, R$ be square matrices of the same size with $Q, R$ Hermitian, and let $J = \left[\begin{smallmatrix} P & -Q \\ -R & -P^\dagger \end{smallmatrix}\right]$ be the associated Hamiltonian matrix, accessed through a block encoding $J/\alpha_J$. Let $\Pi_{\mathrm{in}} = \left[\begin{smallmatrix} \Pi_1 & \Pi_2 \end{smallmatrix}\right]$ be the projector onto its antistable invariant subspace, with compressions $J_{\mathrm{in}} = J\Pi_{\mathrm{in}}$, $J_{\mathrm{out}} = J\Pi_{\mathrm{out}}$, and let $\sigma > 0$ be a lower bound on $\operatorname{\mathbf{Sval}}_{\min}(\Pi_2)$.

Assume the compressions $J_{\mathrm{in}}$ and $J_{\mathrm{out}}$ admit spectral sets \eqref{eq:app-spectral-set} with constant $c$ and gap $\delta$ satisfying \eqref{eq:app-halfplane-cor-hyp} with $\alpha_J$ in place of $\alpha_A$. For any $\varepsilon > 0$, the stabilizing solution $X = -\Pi_2^{+}\Pi_1$ of the continuous-time algebraic Riccati equation $XQX - XP - P^\dagger X - R = 0$ can be block encoded with normalization factor
\begin{equation}
  \label{eq:app-riccati-normalization}
  \alpha_X = \operatorname{\mathbf{O}}\!\left(\frac{c\,\norm{\Pi_{\mathrm{in}}}}{\sigma}\right)
\end{equation}
and accuracy $\varepsilon$, using
\begin{equation}
  \label{eq:app-riccati-queries}
  \operatorname{\mathbf{O}}\!\left(\frac{c^2\,\alpha_J\,\norm{\Pi_{\mathrm{in}}}^2}{\sigma\,\delta}
  \polylog\frac{c\,\alpha_J\,\norm{\Pi_{\mathrm{in}}}}{\sigma\,\delta\,\varepsilon}\right)
\end{equation}
queries to the block encoding $J/\alpha_J$.
\end{corollary}

\begin{proof}
By hypothesis, \cor{app-halfplane} applies to $J$. It block encodes the projector $\Pi_{\mathrm{in}}$ with normalization factor $\alpha_\Pi = \operatorname{\mathbf{O}}(c\,\norm{\Pi_{\mathrm{in}}})$ and accuracy $\varepsilon'$, using
\begin{equation}
  \operatorname{\mathbf{O}}\!\left(\frac{c\,\alpha_J\,\norm{\Pi_{\mathrm{in}}}}{\delta}
  \polylog\frac{c\,\alpha_J\,\norm{\Pi_{\mathrm{in}}}}{\delta\,\varepsilon'}\right)
\end{equation}
queries to $J/\alpha_J$. Restricting the input register to the first or the second column block gives block encodings of $\Pi_1/\alpha_\Pi$ and $\Pi_2/\alpha_\Pi$, with no additional queries.

The nonzero singular values of $\Pi_2/\alpha_\Pi$ are at least $\sigma/\alpha_\Pi$. So the pseudoinverse version of quantum singular value transformation~\cite{gilyen2019} block encodes $\Pi_2^{+}$ with normalization factor $\operatorname{\mathbf{O}}(1/\sigma)$ and accuracy $\varepsilon/2$. This uses $\operatorname{\mathbf{O}}(\frac{\alpha_\Pi}{\sigma}\log\frac{1}{\varepsilon})$ applications of the block encoding of $\Pi_2/\alpha_\Pi$. \lem{app-be-product} then multiplies by $\Pi_1/\alpha_\Pi$ and absorbs the sign into a phase. The result is a block encoding of $X = -\Pi_2^{+}\Pi_1$ with normalization factor $\operatorname{\mathbf{O}}(\alpha_\Pi/\sigma)$, which is \eqref{eq:app-riccati-normalization}.

In total, the projector block encoding is applied $\operatorname{\mathbf{O}}(\frac{\alpha_\Pi}{\sigma}\log\frac{1}{\varepsilon})$ times. Errors in block encodings add up at most linearly. Taking $\varepsilon' = \operatorname{\pmb{\Theta}}(\varepsilon\sigma/(\alpha_\Pi\log\frac{1}{\varepsilon}))$ therefore keeps the final accuracy at $\varepsilon$. The argument of every logarithm in the cost is then polynomial in $c$, $\alpha_J$, $\norm{\Pi_{\mathrm{in}}}$, $1/\sigma$, $1/\delta$, and $1/\varepsilon$, so each logarithm is $\operatorname{\mathbf{O}}(\log\frac{c\,\alpha_J\,\norm{\Pi_{\mathrm{in}}}}{\sigma\,\delta\,\varepsilon})$. Multiplying the number of applications by the cost of each application gives \eqref{eq:app-riccati-queries}. 
\end{proof}

The lower bound $\operatorname{\mathbf{Sval}}_{\min}(\Pi_2) \geq 1/\sqrt{1 + \norm{X}^2}$~\cite{RZL2026} does not depend on the gap. Taking $\sigma = 1/\sqrt{1 + \norm{X}^2}$, the normalization factor \eqref{eq:app-riccati-normalization} becomes $\alpha_X = \operatorname{\mathbf{O}}(c\,\norm{\Pi_{\mathrm{in}}}\sqrt{1 + \norm{X}^2})$, which does not depend on $\delta$. In comparison, the resolvent-integral construction~\cite{RZL2026} block encodes the same projector with normalization factor
$\alpha_{\mathrm{res}} = \operatorname{\mathbf{O}}(\norm{\Pi_{\mathrm{in}}}/\delta)$ in place of $\alpha_\Pi$. This introduces a factor of $1/\delta$ to both the number of pseudoinverse applications and the solution normalization. Together with the $1/\delta$ cost of each projector application, the overall gap dependence is $1/\delta^3$. Preconditioning removes the two factors driven by the normalization, leaving only $1/\delta$.

\subsection{Sylvester equation}
\label{append:app-sylvester}

Consider the Sylvester equation $AX + XB = C$, where $A, B, C$ are square matrices of the same size. Its augmented matrix is $K = \left[\begin{smallmatrix} A & C \\ 0 & -B \end{smallmatrix}\right]$, whose spectrum is the union of the spectra of $A$ and $-B$. The projector onto the invariant subspace of the $A$ branch, along that of the $-B$ branch, is $\Pi_{\mathrm{in}} = \left[\begin{smallmatrix} I & X \\ 0 & 0 \end{smallmatrix}\right]$. The solution $X$ is therefore the top-right block of $\Pi_{\mathrm{in}}$, and no pseudoinverse is needed to extract it.

\begin{corollary}[Sylvester equation]
\label{cor:app-sylvester}
Let $A, B, C$ be square matrices of the same size, and let $K = \left[\begin{smallmatrix} A & C \\ 0 & -B \end{smallmatrix}\right]$ be the associated augmented matrix, accessed through a block encoding $K/\alpha_K$. Suppose the spectra of $A$ and $-B$ lie in the right and left half-planes, respectively. Let $\Pi_{\mathrm{in}} = \left[\begin{smallmatrix} I & X \\ 0 & 0 \end{smallmatrix}\right]$ be the projector onto the right-half-plane invariant subspace of $K$, with compressions $K_{\mathrm{in}} = K\Pi_{\mathrm{in}}$, $K_{\mathrm{out}} = K\Pi_{\mathrm{out}}$.

Assume the compressions $K_{\mathrm{in}}$ and $K_{\mathrm{out}}$ admit spectral sets \eqref{eq:app-spectral-set} with constant $c$ and gap $\delta$ satisfying \eqref{eq:app-halfplane-cor-hyp} with $\alpha_K$ in place of $\alpha_A$. For any $\varepsilon > 0$, the solution $X$ of the Sylvester equation $AX + XB = C$ can be block encoded with normalization factor
\begin{equation}
  \label{eq:app-sylvester-normalization}
  \alpha_X = \operatorname{\mathbf{O}}(c\,\norm{\Pi_{\mathrm{in}}})
\end{equation}
and accuracy $\varepsilon$, using
\begin{equation}
  \label{eq:app-sylvester-queries}
  \operatorname{\mathbf{O}}\!\left(\frac{c\,\alpha_K\,\norm{\Pi_{\mathrm{in}}}}{\delta}
  \polylog\frac{c\,\alpha_K\,\norm{\Pi_{\mathrm{in}}}}{\delta\,\varepsilon}\right)
\end{equation}
queries to the block encoding $K/\alpha_K$.
\end{corollary}

\begin{proof}
By hypothesis, \cor{app-halfplane} applies to $K$. It block encodes the projector $\Pi_{\mathrm{in}}$ with normalization factor $\alpha_\Pi = \operatorname{\mathbf{O}}(c\,\norm{\Pi_{\mathrm{in}}})$ and accuracy $\varepsilon$, using the query count \eqref{eq:app-sylvester-queries}. Restricting the output register to the first block and the input register to the second block extracts the top-right block $X$, with no additional queries. This gives a block encoding of $X$ with the same normalization factor and accuracy, which is \eqref{eq:app-sylvester-normalization}. The special cases $c = 1+\sqrt{2}$ and $c = \kappa_V$ carry over from \cor{app-halfplane}.
\end{proof}

The normalization factor \eqref{eq:app-sylvester-normalization} is close to optimal. From the block form of $\Pi_{\mathrm{in}}$,
\begin{equation}
  \Pi_{\mathrm{in}} \Pi_{\mathrm{in}}^\dagger
  = \begin{bmatrix} I + X X^\dagger & 0 \\ 0 & 0 \end{bmatrix},
\end{equation}
so $\norm{\Pi_{\mathrm{in}}} = \sqrt{1 + \norm{X}^2}$. Any block encoding of $X$ has normalization factor at least $\norm{X}$, so $\alpha_X = \operatorname{\mathbf{O}}(c\sqrt{1 + \norm{X}^2})$ is within a factor $\operatorname{\mathbf{O}}(c)$ of $\max(1, \norm{X})$. In comparison, the direct augmented method~\cite[Remark 6.4]{WL2026} block encodes the same projector through the resolvent integral. Its normalization factor is $\alpha_{\mathrm{res}} = \operatorname{\mathbf{O}}(\norm{\Pi_{\mathrm{in}}}/\delta)$ by \prop{app-resolvent-bound}. Preconditioning removes this factor of $1/\delta$ while keeping the query complexity essentially unchanged.

%%%%%%%%%%%%%%%%%%%%%%%%%%%%%%%%%%%%%%%%%%%%%%%%%%%%%%%%%%%%%%%%%%%%%%%%%%%%%%
\clearpage
\bibliographystyle{myhamsplain2}
\bibliography{oblique_projection.bib}

\end{document}